\documentclass[10pt]{article}
\usepackage[pdftex]{graphicx}
\graphicspath{{./pdf/}}
\usepackage{amsmath,amssymb,amsthm}
\newtheorem{theorem}{Theorem}
\newtheorem{lemma}{Lemma}
\usepackage{floatflt}
\usepackage{wrapfig}
\usepackage{enumerate,cite}
\usepackage{import}
\usepackage{sectsty}
\usepackage[linesnumbered,vlined,ruled]{algorithm2e}
\usepackage[list=true]{subcaption}
\usepackage{csagh}
\usepackage{setspace}
\usepackage{array}
\usepackage{color}

\begin{document}
\begin{opening}

\title{Arbitrary pattern formation by asynchronous opaque robots on infinite grid}
\author[Gayeshpur Government Polytechnic, Department of Science and Humanities, Kalyani, West Bengal - 741234, India, ORCID iD: https://orcid.org/0000-0003-4179-8293, email: manashkrkundu.rs@jadavpuruniversity.in]{Manash Kumar Kundu}
\author[Jadavpur University, Department of Mathematics, Kolkata , West Bengal - 700032, India, ORCID iD: https://orcid.org/0000-0002-0546-3894, email: pritamgoswami.math.rs@jadavpuruniversity.in]{Pritam Goswami}
\author[Jadavpur University, Department of Mathematics, Kolkata , West Bengal - 700032, India, ORCID iD: https://orcid.org/0000-0003-1747-4037, email: satakshighosh.math.rs@jadavpuruniversity.in]{Satakshi Ghosh}
\author[Jadavpur University, Department of Mathematics, Kolkata , West Bengal - 700032, India, email: buddhadeb.sau@jadavpuruniversity.in]{Buddhadeb Sau}

% \begin{abstract}
%   
% \end{abstract}
% 
% \keywords{...}
\begin{abstract}
Arbitrary pattern formation ($\mathcal{APF}$) by mobile robots is studied by many in literature under different conditions and environment. Recently it has been studied on an infinite grid network but with full visibility. In opaque robot model, circle formation on infinite grid has also been studied. In this paper, we are solving $\mathcal{APF}$ on infinite grid with asynchronous opaque robots with lights. The robots do not share any global co-ordinate system. The main challenge in this problem is to elect a leader to agree upon a global co-ordinate where the vision of the robots are obstructed by other robots. Since the robots are on a grid, their movements are also restricted to avoid collisions. In this paper, the aforementioned hardness are overcome to produce an algorithm that solves the problem. 
% Note, BibTex with cs-agh.bst must be used for formatting your references!

\end{abstract}


\end{opening}

\section{Introduction}

Arbitrary Pattern Formation ($\mathcal{APF}$) Problem is a classical problem in swarm robotics which deals with fundamental coordination problem of multi-robot systems. The problem is to design an algorithm which will be used by each autonomous mobile robot of a robot swarm that will guide the robots to form any specific pattern that is given to the robots initially as input. In this problem, the robots are modeled as \textit{autonomous} (there is no central control), \textit{anonymous} (the robots do not have any unique identifiers), \textit{homogeneous} (all robots execute the same algorithm) and \textit{identical} (robots are indistinguishable by their appearance). All the robots can freely move on the plane. Each robot has sensing capability by which they can perceive the location of other robots on the plane. The robots do not have any global coordinate system (each robot has its own local coordinate system) and they operate in \textsc{Look-Compute-Move} (\textit{LCM}) cycles. In the \textsc{Look} phase, a robot takes a snapshot of its surroundings. In the \textsc{Compute}, phase a robot process the information got from the \textsc{Look} phase and in the \textsc{Move}, phase a robot moves to another position (a robot also might stay still in this phase) depending on the output of the \textsc{Compute} phase. 
\subsection{Earlier Works}
The problem of Arbitrary Pattern Formation was first introduced in \cite{Suzuki96distributedanonymous}
and after that this problem has been studied many times in the literature
% (\cite{YAMASHITA20102433,inproceedings,VAIDYANATHAN2021104699, BramasT16,BramasT18,CiceroneSN19,CiceroneSN19a,0001FSY15,DieudonnePV10, FelettiMP18,FlocchiniPSW08,FujinagaYOKY15,LukovszkiH14,BoseKAS21,abs-1910-02706 }).
(\cite{YAMASHITA20102433,inproceedings,VAIDYANATHAN2021104699,BramasT16,BramasT18,CiceroneSN19,CiceroneSN19a,0001FSY15,DieudonnePV10,FelettiMP18,FlocchiniPSW08,FujinagaYOKY15,LukovszkiH14,BoseKAS21,abs-1910-02706}).
Initially the problem has been studied only assuming that the robots do not have obstructed visibility. But in a more practical setting, when more than two robots are in a straight line, a robot with camera sensors can only see its adjacent robots (at most two robots). These robots are known as opaque robots. In \cite{BoseKAS21}, $\mathcal{APF}$ has been solved considering opaque robots with visible lights which can assume  6 persistent colors. They have also assumed one axis agreement for each robots. This model of luminous robots has first been introduced in \cite{Peleg05} by Peleg et al. The visible lights can be used by robots as a means of communication and persistent memory. In \cite{BoseKAS21}, the robots are considered to be point robots. But in real life scenario, robots are physical entities that have certain dimensions. So in \cite{abs-1910-02706}, the problem of $\mathcal{APF}$ has been considered and solved using fat robots with luminous robots having 10 colors.

In a plane, the robots can move freely in any direction. So collision can be avoided by the robots by comparably easy techniques. So solving $\mathcal{APF}$ in specific network (eg. grid) is quite interesting itself where it is not so easy to avoid collisions. In \cite{BoseAKS20}, the authors have considered this problem and produced an algorithm with robots having full visibility on an infinite grid in $\mathcal{OBLOT}$ model. Furthermore, in obstructed visibility model, a problem of circle formation on an infinite grid has been solved in \cite{AdhikaryKS21}  with opaque luminous robots with 7 colors. In this paper, we are considering the problem of arbitrary pattern formation on an infinite grid with luminous robots with obstructed visibility. 

\subsection{Problem description and our contribution}
This paper deals with the arbitrary pattern formation problem on an infinite grid using opaque luminous robots with 8 colors. The robots operate in \textit{LCM} cycles under an adversarial asynchronous scheduler.
The robots are autonomous, anonymous, identical and homogeneous. They move only through the edges of the grids and the movement is instantaneous for each robot (i.e a robot can only be seen on a grid point). Initially the robots are placed arbitrarily on the grid. From this configuration, they need to move to a target configuration or, Pattern (a set of target coordinates) without collision. The robots have one axis agreement and does not have agreement on global coordinate (each robot has its own local coordinate). 

The main difficulty of the problem is visibility. As $\mathcal{APF}$ is closely related to the \textsc{Leader Election} problem, without seeing the whole configuration it is quite hard to elect a leader and thus design an algorithm to solve the problem. Depending on the local view of each robot, the algorithm is need to be designed. In this paper, the described algorithm does so. Also another difficulty was to avoid collision between robots while they are moving. We removed this difficulty by using a technique where the robots will move sequentially.

The problem described in this paper is also very practical in nature. The restricted movement and also the obstructed visibility, these practical scenarios are considered here. The algorithm described in this paper solves the above mentioned $\mathcal{APF}$ problem in  total $\mathcal{O}(kD)$ moves in the worst case, where $k$ is the number of robots on the grid and $D$ is $max\{m,n,M,N,k\}$ ($m$, $n$ are the height and width of smallest enclosing rectangle of the initial configuration; $M$, $N$ are the same for the target configuration).

\section{Model and Definitions}\label{model}

\subsection{Model}

~~~~~\textbf{Robots:} Robots are autonomous, anonymous, homogeneous and identical. They are deployed on a two-dimensional infinite grid where each of them is initially positioned on distinct grid points. They do not have a common notion of direction. The robots have an agreement over the positive direction of X-axis i.e, all the robots have an agreement over left and right. They do not have any agreement over the Y-axis. Here the robots do not have access to any global coordinate system other than the agreement over the positive direction of X-axis. The total number of robots is not known to them. The robots are assumed to be dimensionless and modeled as points.

% \begin{floatingfigure}[r]{7.0cm}

%   \fontsize{7pt}{7pt}\selectfont
%   \def\svgwidth{0.4\textwidth}
%   \import{}{pdf/stack1.pdf_tex}
%   \caption{Illustrations for the geometric definitions given in Section \ref{model}.}
%   \label{stack}

% \end{floatingfigure}

\textbf{Look-Compute-Move cycles:} The robot, when active, operates according to the \textsc{Look-Compute-Move} cycle. In the \textsc{Look} phase, a robot takes the snapshot of the positions of all the robots represented in its own local co-ordinate system. Then the robot performs computation and compute the next position and a light according to a deterministic algorithm i.e., the \textsc{Compute} phase. In the \textsc{Move} phase, it will either move unit length to the desired location along a straight line or make a null move.

\textbf{Scheduler:} We assume that the robots are controlled by an asynchronous adversarial scheduler. This implies that the amount of time spent in \textsc{Look}, \textsc{Compute}, \textsc{Move}, and inactive states by different robots is finite but unbounded and unpredictable. As a result, the robots do not have a common notion of time, and the configuration perceived by a robot during the \textsc{Look} phase may significantly change before it actually makes a move.

\textbf{Movement:} The movement of robots are restricted only along grid lines from one grid point to one of its four neighboring grid points. Robots' movements are assumed to be instantaneous in discrete domains. Here we assume that the movements are instantaneous. The robots are always seen on grid points, not on edges.

\textbf{Visibility:} The robots visibility is unlimited but by the presence of other robots it can be obstructed. A robot $r_i$ can see another robot $r_j$ if and only if there are no robots on the straight line segment $\overline{r_i r_j}$.

\textbf{Lights:} Each robot is equipped with an externally visible light, which can assume a $\mathcal{O}(1)$ number of predefined lights. The robots communicate with each other using these lights. The lights are not deleted at the end of a cycle, but otherwise, the robots are oblivious. The lights used in our algorithm are $\{\texttt{off}$, $\texttt{terminal1}$, $\texttt{symmetric}$, $\texttt{decider}$, $\texttt{call}$, $\texttt{leader1}$, $\texttt{leader}$, $\texttt{done} \}$.
\subsection{Notations and Definitions}
We have used some notations throughout the paper. A list of these notations along with their definitions are mentioned in the following table.

\vspace{0.01\linewidth}
\begin{center}
\begin{tabular}{ | m{4em} | m{10cm}| } 
\hline
$\mathcal{L}_1$& First vertical line on left that contains at least one robot.\\
  \hline
     $\mathcal{L}_V(r)$ & The vertical line on which the robot $r$ is located.\\
    \hline
    $\mathcal{L}_H(r)$ & The horizontal line on which the robot $r$ is located.\\
    \hline
     $\mathcal{L}_I(r)$ & The left immediate vertical line of robot $r$ which has at least one robot on it. \\
   \hline
     $\mathcal{R}_I(r)$&  The right immediate vertical line of robot $r$ which has at least one robot on it. \\
    \hline
    ${H}_L^O(r)$ &  Left open half for the robot $r$. \\
    \hline
     ${H}_L^C(r)$ & Left closed half for the robot $r$ (i.e  ${H}_L^O(r) \cup \mathcal{L}_V(r)$).\\
     \hline
     ${H}_B^O(r)$ &  Bottom open half for the robot $r$.\\
     \hline
     ${H}_B^C(r)$ &  Bottom closed half for the robot $r$ (i.e  ${H}_B^O(r) \cup \mathcal{L}_H(r)$).\\
     \hline
     ${H}_U^O(r)$ & Upper open half for the robot $r$.\\
    \hline
    ${H}_U^C(r)$ & Upper closed half for the robot $r$ (i.e  ${H}_U^O(r) \cup \mathcal{L}_H(r)$). \\
    \hline
    $L_{t_j-1}$ &   The horizontal line below the target position $t_j$. \\
    \hline
    $K$ & The horizontal line passing through the middle point of the line segment between two robots with light \texttt{decider} or \texttt{terminal1} on the same vertical line. \\
    \hline
    $L_{H1}$ & The immediate horizontal line above the robot with light \texttt{leader}. \\
      \hline
\end{tabular}
\end{center}
\vspace{0.01\linewidth}
Some additional definitions are needed to be explained which will be useful later.

\textit{\textbf{Configuration:}} Let us consider a team of robots placed on an simple undirected connected graph $G = (V,E)$. Let us define a function $f:V \rightarrow \{0\} \cup \mathcal{N}$, where $f(v)$ is the number of robots placed on vertex $v$. The graph $G$ together with the function $f$ is called a configuration which is denoted by $ \mathbb{C}=(G,f)$. For ant time $T$, $\mathbb{C}(T)$ will denote the configuration of the robots at time $T$.

For a graph $G =(V, E)$, $\phi:V \rightarrow V$ is an automorphism if $\phi$ is a bijection and $\phi(u) \phi(v)$ is adjacent iff $u$ and $v$ are adjacent $\forall u,v \in V$. All the automorphisms of $G$ form a group denoted by $Aut(G)$. Similarly we can define an automorphism $\phi$ for a configuration $(G,f)$ where $\phi \in Aut(G)$ and $f(u) = f(\phi(u)), \forall u \in V$. All automorphisms on $(G,f)$ form a group denoted by $Aut(G,f)$.

\textit{\textbf{Symmetric configuration:}} For any configuration $ \mathbb{C}=(G,f)$, we can define the group $Aut(\mathbb{C})$. $\phi(v) = v, \forall v \in V$ is called a trivial symmetry. Every  non trivial $\phi \in Aut(\mathbb{C})$ is called a symmetry of $\mathbb{C}$.  Note that all symmetric configurations  of a configuration $\mathbb{C}$ is basically generated by some translations, rotations and reflections. Translation shifts all the vertices by the same amount. Since the number of robots in the configuration $\mathbb{C}$ is finite it is easy to see that there is no translation in $Aut(\mathbb{C})$. Reflections are defined by some axis or line of reflection. It can be vertical, horizontal or diagonal. The angle of rotation can be 90$^\circ$ or 180$^\circ$. The center of rotation can be a vertex of the grid, center of an unit square or a center of an edge.

\textit{\textbf{Stable Configuration:}} A configuration $\mathbb{C}$ is called a stable configuration if the following conditions are satisfied in $\mathbb{C}$.
\begin{center}
    \begin{enumerate}
        \item There are two robots with light \texttt{decider} on same vertical line and all other robots in $\mathbb{C}$ have light \texttt{off}.
        \item The vertical and horizontal line on which the robots with light \texttt{decider} are located don't have any other robots.
        \item The robots with light \texttt{decider} have no robots on left open half  and also their upper closed half or bottom closed half have no other robots.
    \end{enumerate}
\end{center}

\textit{\textbf{Leader Configuration:}}  A configuration $\mathbb{C}$ is called a leader configuration if the following conditions are satisfied in $\mathbb{C}$.
\begin{center}
    \begin{enumerate}
        \item There are exactly one robot with light \texttt{leader} and all other robots have light \texttt{off}.
        \item The vertical line and the horizontal line on which the robot with light \texttt{leader} is located do not have other robots.
        \item The robots with light \texttt{leader} has no robots on  left open half and also upper open half or bottom open half is empty .
    \end{enumerate}
\end{center}

\textit{\textbf{Compact Line:}} A line is called compact if there is no unoccupied grid position between any two robots on that line.

\textit{\textbf{Terminal Robot:}} A robot $r$ is called a terminal robot if $\mathcal{L}_V(r) \cap H$ is empty, where $H \in \{H_B^O(r), H_U^O(r)\}$.

 \textit{\textbf{Symmetry of a vertical line $L$ w.r.t $K$:}} Let $L$ be a vertical line of the grid and $\lambda$ be a binary sequence defined on $L$ such that $j$-th term of $\lambda$ is defined as follows:
 \begin{equation*}
 \lambda(j) = \begin{cases}
       1 & \text{if $\exists$ a robot on the $j$-th grid point from $K \cap L$ on the line $L$.} \\
       0 & \text{otherwise.}
     \end{cases}
\end{equation*}
By definition of $\lambda$, it follows that there are two such values of $\lambda$, say $\lambda_1$ and $\lambda_2$. We say that the line $L$ is symmetric with respect to $K$ if $\lambda_1 = \lambda_2$. For future, whenever symmetry of a line is mentioned, it is assumed that it means the symmetry of the line with respect to $K$.

 %let $\lambda_1$ and $\lambda_2$ be two binary sequences such that $j$-th term of $\lambda_1$ and $\lambda_2$ are defined as follows:
 %\begin{enumerate}
   %  \item $\lambda_1(j) = 1$, if $\exists$ a robot on the $j$-th grid point on the line $L$ above and starting from $K \cap L$ and 0 otherwise.
    % \item $\lambda_2(j) = 1$, if $\exists$ a robot on the $j$-th grid point on the line $L$ below and starting from $K \cap L$ and 0 otherwise.
 %\end{enumerate}

 \textit{\textbf{Dominant half:}} A robot $r$ is said to be on the dominant half if the following conditions are satisfied:
 \begin{enumerate}
     \item $\mathcal{R}_I(r)$ is not symmetric with respect to $K$.
     \item If  lexicographically $\lambda_1 > \lambda_2$ on $\mathcal{R}_I(r)$, then $r$ and the portion of $\mathcal{R}_I(r)$ corresponding to $\lambda_1$ lie on same half plane delimited by $K$.
 \end{enumerate}
\section{The Algorithm}
The main result of the paper is Theorem 1. The proof of the ‘only if’ part is the same as in case for point
robots, proved in \cite{BoseKAS21}. The ‘if’ part will follow from the algorithm presented in this section.

\begin{theorem}
\label{thm1.1}
For a set of opaque luminous robots having one axis agreement, $\mathcal{APF}$ is deterministically solvable if and only if the initial configuration is not symmetric with respect to a line $K$ such that 1) $K$ is parallel to the agreed axis and 2) $K$ is not passing through any robot.
\end{theorem}

For the rest of the paper, we shall assume that the initial configuration $\mathbb{C}(0)$ does not admit the
unsolvable symmetry stated in Theorem \ref{thm1.1}. Our algorithm works in two stages. In the first stage, Leader Election, the robots will agree on a leader. Since there are no common agreement on a global coordinate system, the robots will not be able to agree on the embedding of the pattern on the grid. Thus leader election is necessary for robots to agree on a global coordinate. This stage is further divided in two phases namely \textit{Phase 1} and \textit{Phase 2}. Next in the pattern formation stage, the robots will move to form the input pattern embedded on the grid. The stages are described in details in \ref{ss1} and \ref{ss2}.
\subsection{Leader Election}
\label{ss1}
\subsubsection{Phase 1}
\begin{figure}[ht!]
  \centering
   \begin{minipage}{1\linewidth}
  \begin{algorithm}[H]
     \setstretch{0.5}
    \SetKwInOut{Input}{Input}
    \SetKwInOut{Output}{Output}
    \SetKwProg{Fn}{Function}{}{}
    \SetKwProg{Pr}{Procedure}{}{}

    \Pr{\textsc{Phase1()}}{

    $r \leftarrow$ myself
    %\\$\mathcal{L}_i$ is a horizontal line with $i \geq 2$.
%     \\$\mathcal{L}_r$ = horizontal line where $r$ resides.\\$\mathcal{L}_1$ = lowest horizontal line.

    %\While{\textsc{LineFormation() = False}}{

    \uIf{$r.light =$ \texttt{off}}
        {
    
     \uIf{$r$ is terminal and there is no robot in \textbf{$H_{L}^{O}(r)$} and no  robot \texttt{leader1} in $\mathcal{R}_I(r)$}{$r.light \leftarrow$ \texttt{terminal1} \\Move left}
     \ElseIf{there are exactly two robots in \textbf{$\mathcal{L}_I(r)$} having light \texttt{terminal1} and $r$ is on $K$}{
         $r.light \leftarrow$ \texttt{leader1} }
         
         }

    \uElseIf{$r.light = $\texttt{terminal1}}
    {
    \uIf{there is a robot with light \texttt{terminal1} on $\mathcal{L}_{V}(r)$ }
        {
        \uIf{no robots on $K \cap \mathcal{R}_I(r)$}{
            \uIf{$\mathcal{R}_I(r)$ is symmetric with respect to $K$ }{$r.light \leftarrow$ \texttt{symmetric}}
    
        \Else{\If{$r$ is in dominant half}{$r.light \leftarrow$ \texttt{leader1}} 
        
        }
                          
        }
        \ElseIf{there is a robot on $K \cap \mathcal{R}_I(r)$ with light \texttt{leader1}}{
            $r.light \leftarrow$ \texttt{off}
        }
         }
     
     \uElseIf{there is a robot with light \texttt{symmetric} on $\mathcal{L}_{V}(r)$}{$r.light \leftarrow$ \texttt{symmetric}}
     
     \uElseIf{there is a robot with light \texttt{leader1} or \texttt{off} on $\mathcal{L}_{V}(r)$}{$r.light \leftarrow$ \texttt{off}}
     
     \ElseIf{$r$ is singleton on $\mathcal{L}_{V}(r)$ and all robots in $\mathcal{R}_I(r)$ are \texttt{off}}{$r.light \leftarrow$ \texttt{leader1}}

    }

    \ElseIf{$r.light = $\texttt{symmetric}}
    {
    \If{there is a robot $r'$ with light \texttt{symmetric} or \texttt{decider} on $\mathcal{L}_{V}(r)$}
        {  \uIf{there is other robot both in $H_{U}^{C}(r)$ and        $H_{B}^{C}(r)$}{move vertically opposite to $r'$}
           \Else{$r.light \leftarrow$ \texttt{decider}}
        }
    }

  }

    \caption{\textbf{Phase 1}}
    \label{Algo_Phase1}
\end{algorithm}

 \end{minipage}
\end{figure}
The procedure \textit{Phase 1} starts from the initial configuration with the aim of forming a stable configuration or achieving a configuration with a robot with light \texttt{leader1}. Initially all the robots are located on an infinite grid with light \texttt{off}. On waking, each robot $r$ will check if the robot is terminal and if their open left half has no other robots and no  robot \texttt{leader1} in $\mathcal{R}_I(r)$. Note that there will be at most two and at least one such robot. These robots will change their lights to \texttt{terminal1} and move left (Figure \ref{Fig:initial}). 

\begin{figure}[ht]
     \includegraphics[width=0.6\linewidth]{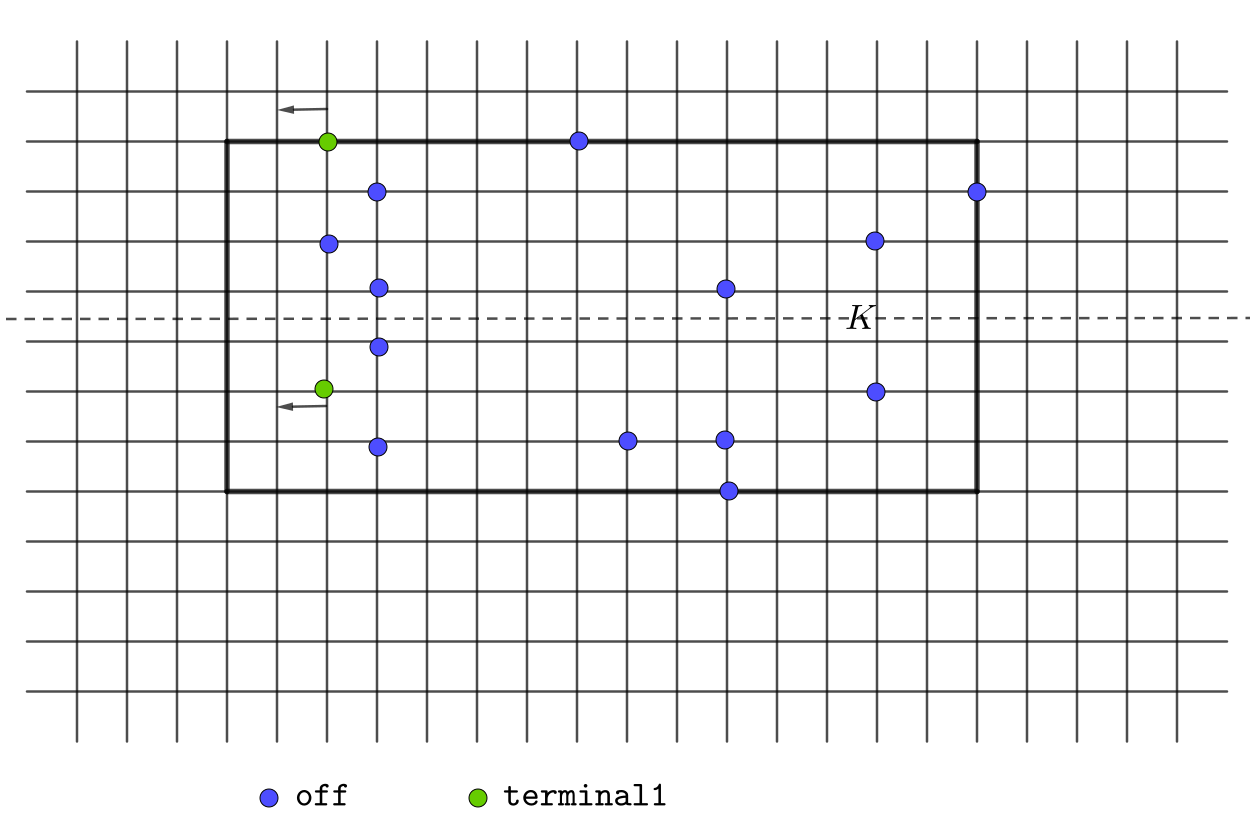}
     \caption{Terminal robots on the line $\mathcal{L}_1$ change lights to \texttt{terminal1} and move left.}\label{Fig:initial}
\end{figure}

If there is only one robot (say, $r$) with light \texttt{terminal1} on $\mathcal{L}_1$ and light of all robots on $\mathcal{R}_I(r)$ are \texttt{off}, then $r$ will change its light to \texttt{leader1} (Figure \ref{Fig:singleTerminal1}, \ref{Fig:singleTerminal12Leader1}). 
This might also happen due to asynchrony of the system that there are two robots with \texttt{terminal1} light but one (say, $r_1$) in $\mathcal{L}_1$ and the other one (say, $r_2$) is still in $\mathcal{R}_I(r_1)$. In this case, the robot on $\mathcal{L}_1$ will see that all robots on $\mathcal{R}_I(r_1)$ do not have lights \texttt{off} and will wait for $r_2$ to reach $\mathcal{L}_1$.

% \begin{figure}[!htb]\begin{center}
%     \includegraphics[width=0.7\linewidth]{pdf/onK.png}
%      \caption{}\label{Fig:onK}
% \end{center}
% \end{figure}
% \begin{figure}[!htb]\centering
%   \begin{minipage}{0.45\textwidth}
%      \includegraphics[width=1.15\linewidth]{pdf/terminal1 asymm.png}
%      \caption{}\label{Fig:initial}
%   \end{minipage}
%   \begin {minipage}{0.45\textwidth}
%     \includegraphics[width=1.15\linewidth]{pdf/terminal1 asymm leader1.png}
%      \caption{}\label{Fig:single robot on line}
%   \end{minipage}
% \end{figure}
% \begin{figure}[!htb]\centering
%   \begin{minipage}{0.45\textwidth}
%      \includegraphics[width=1.15\linewidth]{pdf/terminal1 symm.png}
%      \caption{}\label{Fig:initial}
%   \end{minipage}
%   \begin {minipage}{0.45\textwidth}
%     \includegraphics[width=1.15\linewidth]{pdf/Symmetric.png}
%      \caption{}\label{Fig:single robot on line}
%   \end{minipage}
% \end{figure}

If there are two robots (say, $r_1$ and $r_2$) on $\mathcal{L}_1$ with \texttt{terminal1} light, then observe that both the robots $r_1$ and $r_2$ and all the robots on $\mathcal{R}_I(r_1)$ ($= \mathcal{R}_I(r_2)$) can recognise the line $K$. Now, if there exists a robot (say $r_l$) occupying the grid point $\mathcal{R}_I(r_1) \cap K$, then $r_l$  changes its light to \texttt{leader1} and the robots with light \texttt{terminal1} changes their light to \texttt{off} after seeing the robot $r_l$ with light \texttt{leader1} (Figure \ref{Fig:robotOnK}, \ref{Fig:robotOnK2Leader1}). If the grid point $\mathcal{R}_I(r_1) \cap K$ is empty, then the robots $r_1$ and $r_2$ check the symmetry of the line $\mathcal{R}_I(r_1)$ ($=\mathcal{R}_I(r_2)$). If $\mathcal{R}_I(r_1)$ is not symmetric, then the robot $r_i$($i = 1$ or, $2$), which is on the dominant half changes its light to \texttt{leader1} (Figure \ref{Fig:Terminal1Assymetry}, \ref{Fig:Terminal1Assymetry2Leader1}). On the other hand if $\mathcal{R}_I(r_1)$ is symmetric, then both the robots $r_1$ and $r_2$ change their lights from \texttt{terminal1} to \texttt{symmetric} (Figure \ref{Fig:Terminal1Symm}, \ref{Fig:Terminal1Symm2Symmetry}). Note that, due to asynchrony it might happen that one robot, let's say $r_1$ changes its light to \texttt{symmetric} before $r_2$. Then $r_1$ does not move until it sees another robot ($r_2$) on $\mathcal{L}_V(r_1)$ with light \texttt{symmetric}. This technique prevents the configuration from getting symmetric in this phase. Now after seeing another robot ($r_2$) on $\mathcal{L}_V(r_1)$ with \texttt{symmetric} light, $r_1$  moves vertically opposite of $r_2$ until the closed upper half or the closed bottom half of $r_1$ has no other robots (similar argument can be given for $r_2$) (Figure \ref{Fig:SymmetryMovement}). After reaching their designated positions both the robots $r_1$ and $r_2$ change their light from \texttt{symmetric} to \texttt{decider}, achieving a stable configuration (Figure \ref{Fig:StableConfig}). Note that, due to asynchrony it might happen that one robot, let's say $r_1$ reaches to its designated position and changes its light to \texttt{decider} before $r_2$.

The following Lemmas \ref{lemma:stage1_1}, \ref{lemma:stage1_2} and \ref{lemma:stage1_3} and Theorem \ref{Algo_Phase1} justify the correctness of the Algorithm \ref{Algo_Phase1}.

\begin{lemma}
\label{lemma:stage1_1}
  If the initial Configuration $\mathbb{C}(0)$ has exactly one robot on $\mathcal{L}_1$, then  $\exists$ $T_1 > 0$ such that there will be exactly one robot with light \texttt{leader1} in $\mathbb{C}(T_1)$.  
\end{lemma}
\begin{proof}
 Let us assume the initial configuration $\mathbb{C}(0)$ has exactly one robot (say, $r$) on $\mathcal{L}_1$. According to the Algorithm \ref{Algo_Phase1}, when $r$ wakes it will see that it is a terminal robot, the open left half is empty and no \texttt{leader1} in $\mathcal{R}_I(r)$. Then it  changes the light to \texttt{terminal1} and moves left. Note that, after $r$ moves left, $\mathcal{L}_1$ now denotes the new vertical line where $r$ is located.  Now on waking again, $r$ sees that it is the only robot on $\mathcal{L}_1$ with light \texttt{terminal1} and   all other robots on $\mathcal{R}_I(r)$ having light \texttt{off}. In this case, $r$ changes the light to \texttt{leader1}. Note that all the other robots in this case  have the light \texttt{off} throughout the Algorithm \ref{Algo_Phase1}. So during the execution of Algorithm \ref{Algo_Phase1}, if this case occurs, there will be exactly one robot with light \texttt{leader1} (Figure \ref{Fig:singleTerminal1}, \ref{Fig:singleTerminal12Leader1}).
 
 \begin{figure}[!htb]\centering
   \begin{minipage}{0.45\textwidth}
    \includegraphics[width=1.15\linewidth]{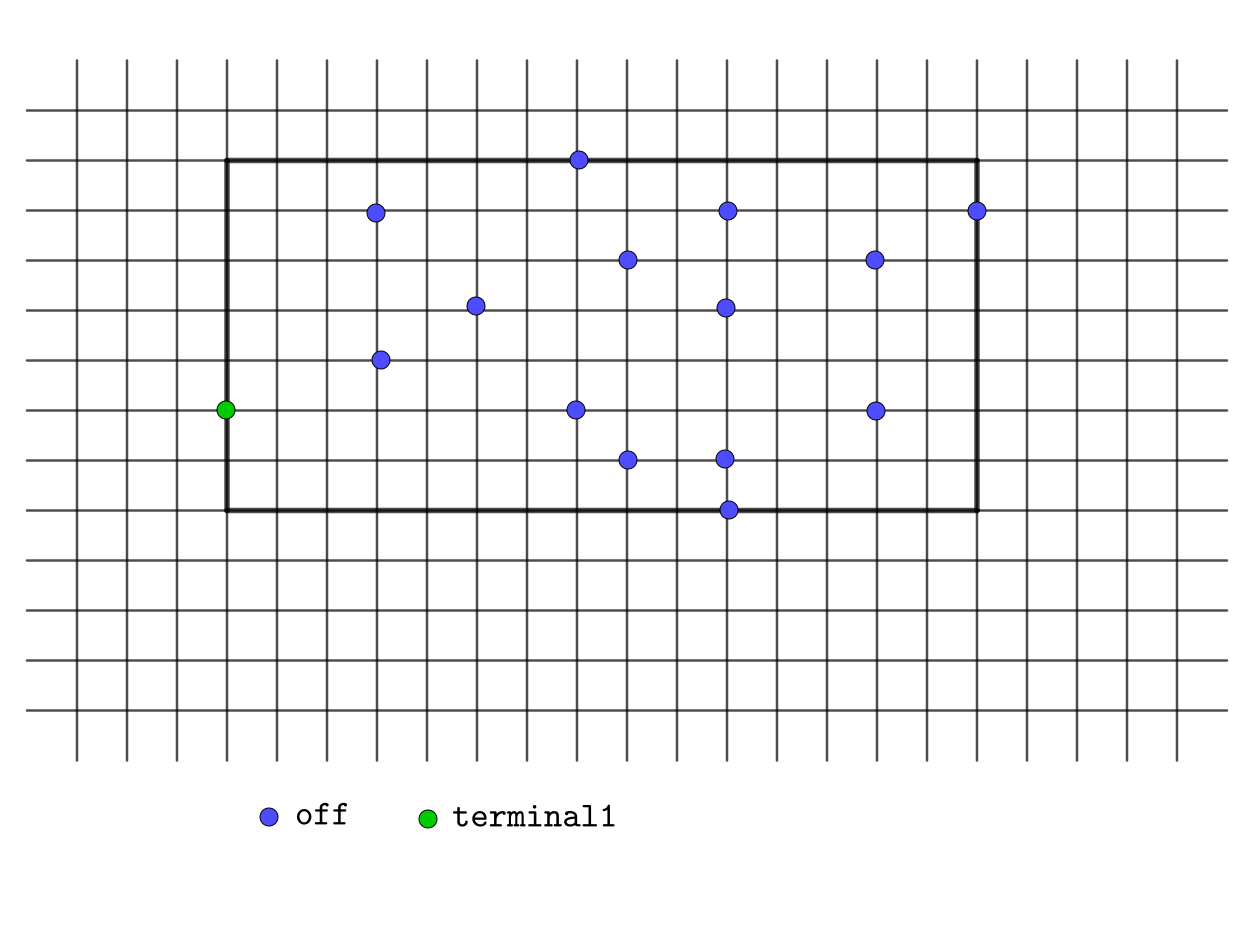}
     \caption{Single robot on $\mathcal{L}_1$ with light \texttt{terminal1}.}\label{Fig:singleTerminal1}
   \end{minipage}
   \hspace{1mm}
   \begin {minipage}{0.45\textwidth}
    \includegraphics[width=1.15\linewidth]{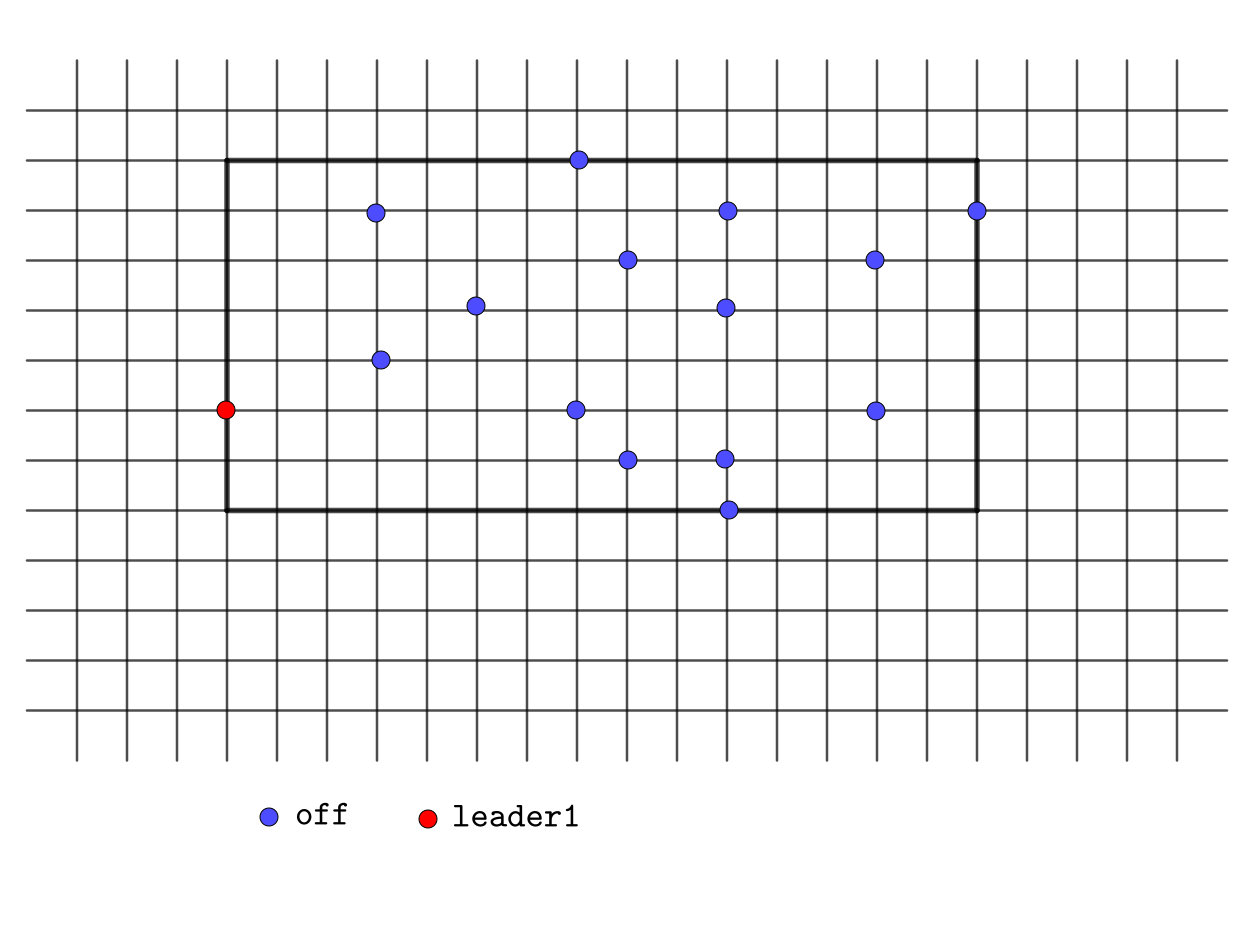}
     \caption{The single robot with light \texttt{terminal1} changes its light to \texttt{leader1}.}\label{Fig:singleTerminal12Leader1}
   \end{minipage}
\end{figure}

\end{proof}
\begin{lemma}
\label{lemma:stage1_2}
  If the initial configuration $\mathbb{C}(0)$ has exactly two robots on $\mathcal{L}_1$, then $\exists$ $T' >0$ such that $\mathbb{C}(T')$ is either a stable configuration or there is exactly one robot with light \texttt{leader1} in  $\mathbb{C}(T')$.
\end{lemma}
\begin{proof}
 Let $r_1$ and $r_2$ be two robots on $\mathcal{L}_1$ in the initial configuration $\mathbb{C}(0)$. In this case, both the robots change their lights to \texttt{terminal1} and move left. Note that after moving left, $\mathcal{L}_1$ now denotes the vertical line on which the robots are located now. Due to asynchrony of the system, two cases may occur:
 
 \textbf{Case 1:} Both the robots $r_1$ and $r_2$ reaches\textbf{ $\mathcal{L}_1$} before waking again. In this case, $r_1$ and $r_2$ can see each other and can calculate the line \textbf{$K$}. Now if there is another robot (say $r_l$) on the intersection of $K$ and $\mathcal{R}_I(r_1)$ ($\mathcal{R}_I(r_1) = \mathcal{R}_I(r_2)$), then $r_l$  changes the light to \texttt{leader1}. Note that $r_l$ can see both $r_1$ and $r_2$, so it can calculate the line $K$ itself. After seeing $r_l$ with light \texttt{leader1}, the robots with light \texttt{terminal1} change their lights to \texttt{off} (Figure \ref{Fig:robotOnK}, \ref{Fig:robotOnK2Leader1}).

\begin{figure}[ht]\centering
   \begin{minipage}{0.45\textwidth}
     \includegraphics[width=1.15\linewidth]{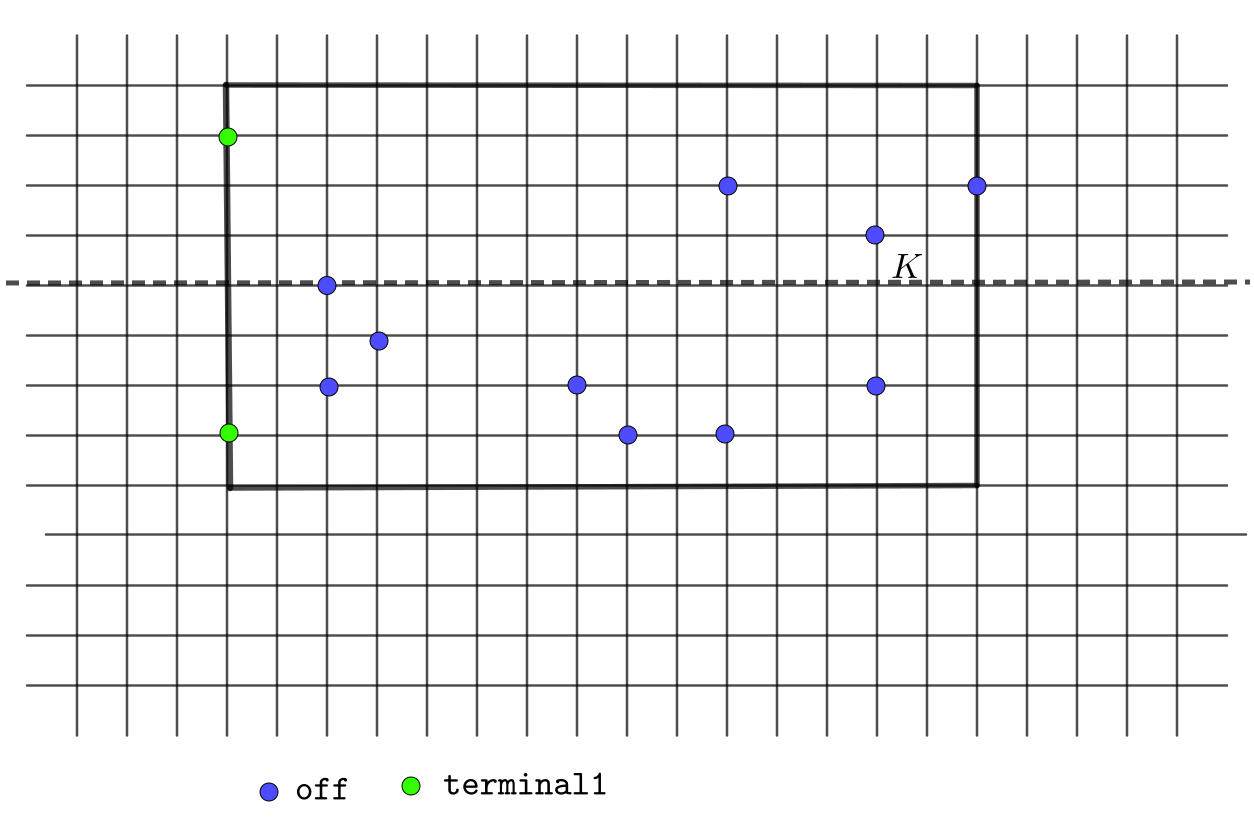}
     \caption{Both the robots with light \texttt{terminal1} see a robot on the right next occupied vertical line and on the line $K$. }\label{Fig:robotOnK}
   \end{minipage}
   \hfill
   \begin {minipage}{0.45\textwidth}
    \includegraphics[width=1.15\linewidth]{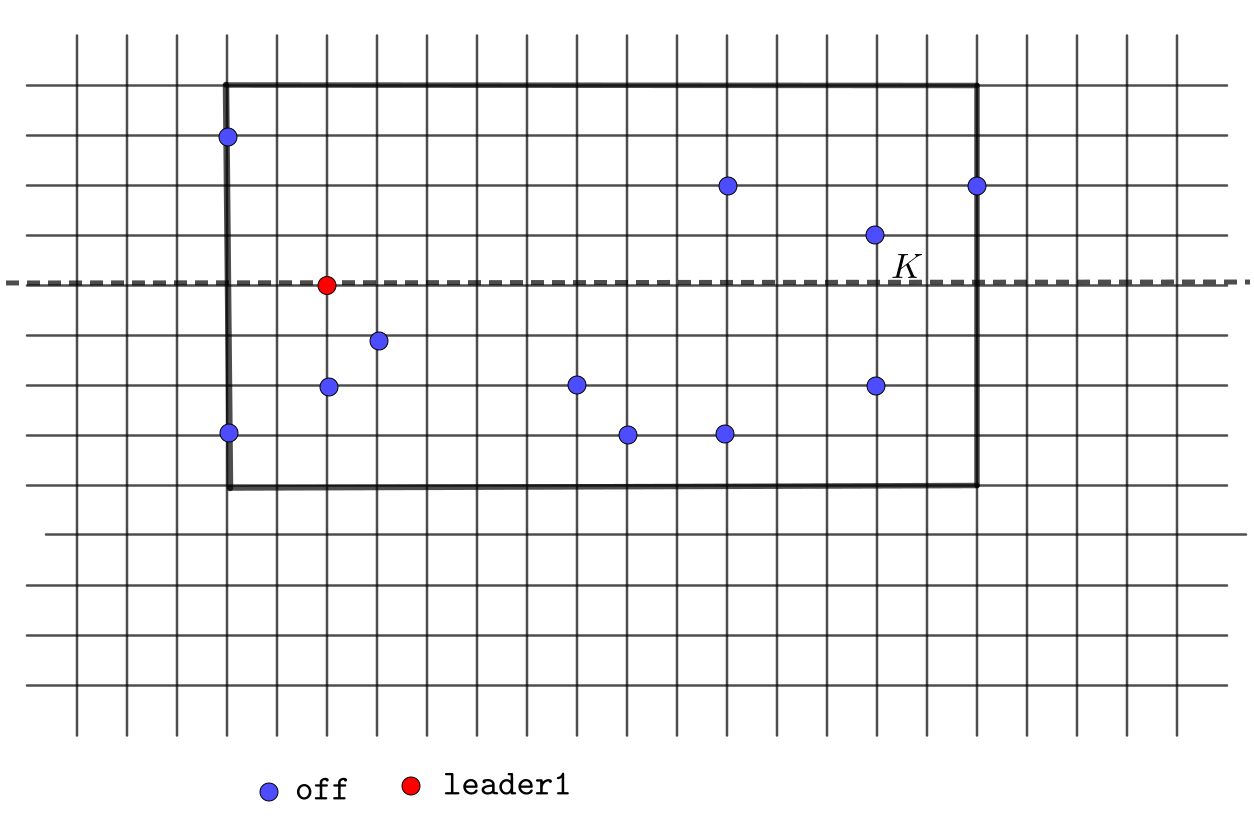}
     \caption{The robot on $K$ changes its light to \texttt{leader1} and next the \texttt{terminal1} robots see robot \texttt{leader1} and change their lights to \texttt{off}.}\label{Fig:robotOnK2Leader1}
   \end{minipage}
\end{figure}
 
  On the other hand, if there is no robot on the intersection of $K$ and $\mathcal{R}_I(r_1)$, then $r_1$ and $r_2$ both check the symmetry of $\mathcal{R}_I(r_1)$ with respect to the line $K$. If $\mathcal{R}_I(r_1)$ is asymmetric with respect to  $K$, then the robot in the dominant half ($r_1$ or $r_2$) will change it's light to \texttt{leader1} (Figure \ref{Fig:Terminal1Assymetry}, \ref{Fig:Terminal1Assymetry2Leader1}). 
  \begin{figure}[ht]\centering
   \begin{minipage}{0.45\textwidth}
     \includegraphics[width=1.15\linewidth]{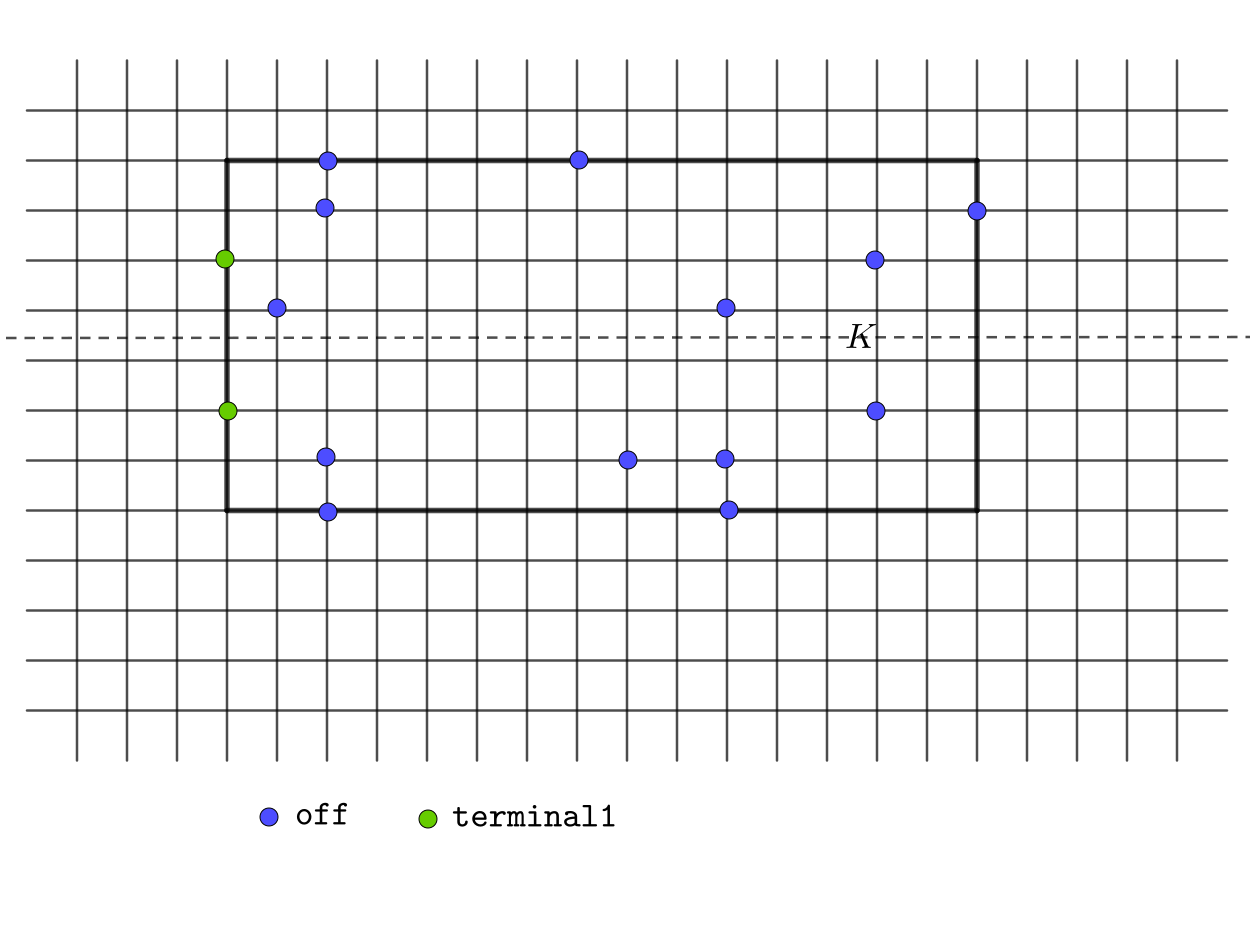}
     \caption{The robots with light \texttt{terminal1} see that the right next occupied vertical line is not symmetric.}\label{Fig:Terminal1Assymetry}
   \end{minipage}
   \hfill
   \begin{minipage}{0.45\textwidth}
    \includegraphics[width=1.15\linewidth]{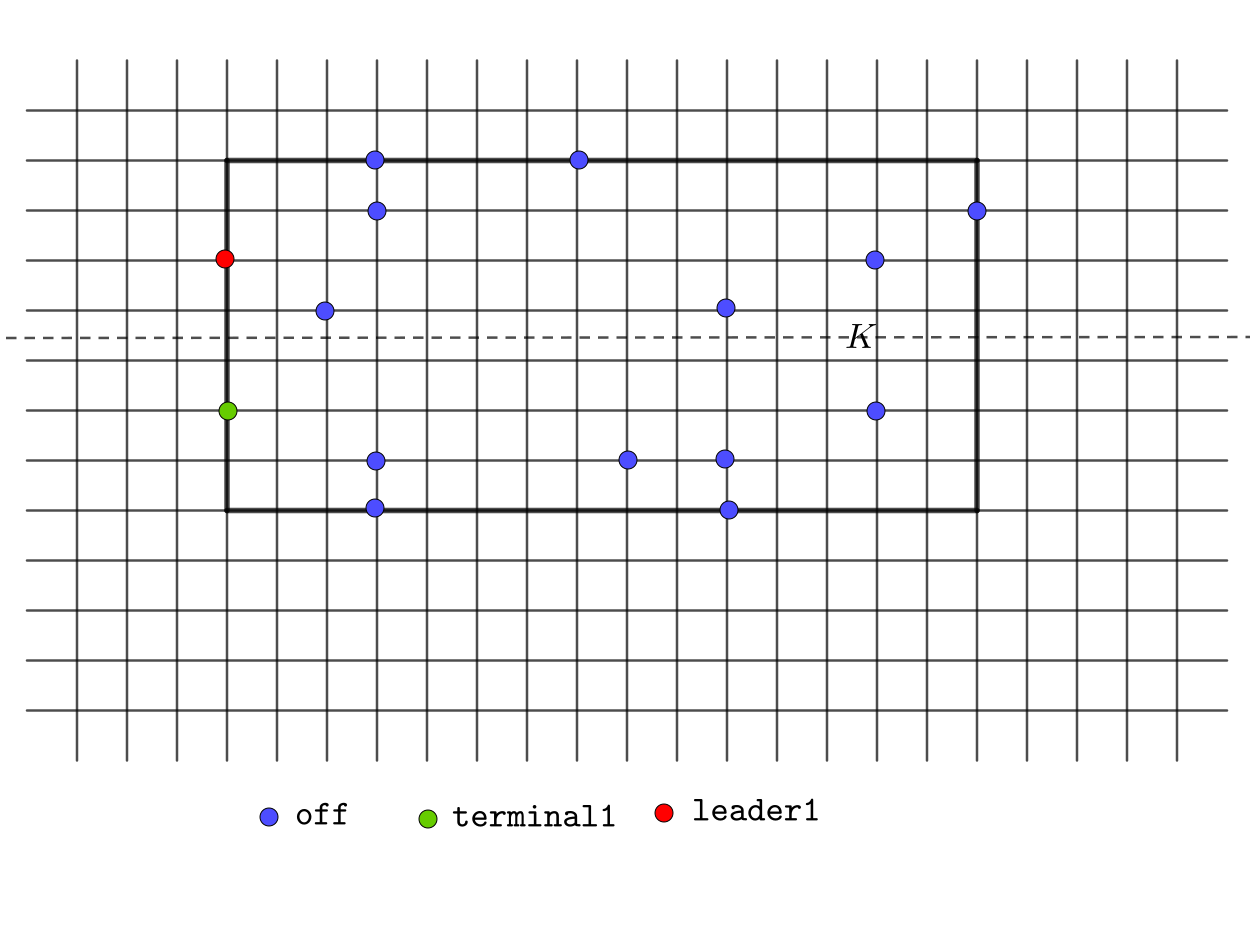}
     \caption{The robot with light \texttt{terminal1} on the dominant half changes the light to \texttt{leader1}.}\label{Fig:Terminal1Assymetry2Leader1}
   \end{minipage}
\end{figure}

\begin{figure}[ht]\centering
   \begin{minipage}{0.45\textwidth}
   \begin{center}
       \includegraphics[width=1.15\linewidth]{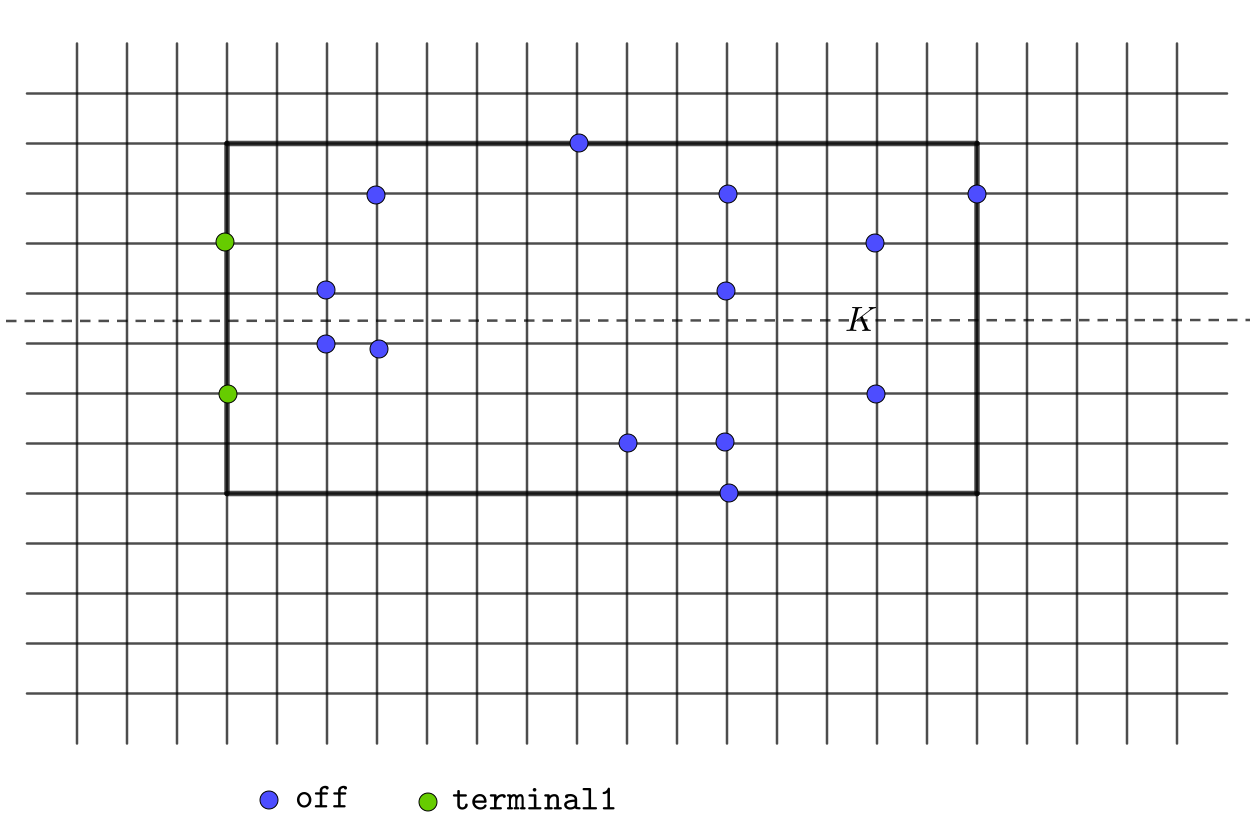}
     \caption{The \texttt{terminal1} robots see that the right next occupied  vertical line is symmetric.}\label{Fig:Terminal1Symm}
     
   \end{center}
 \end{minipage}
   \hspace{1mm}
   \begin{minipage}{0.45\textwidth}
   \begin{center}
       \includegraphics[width=1.15\linewidth]{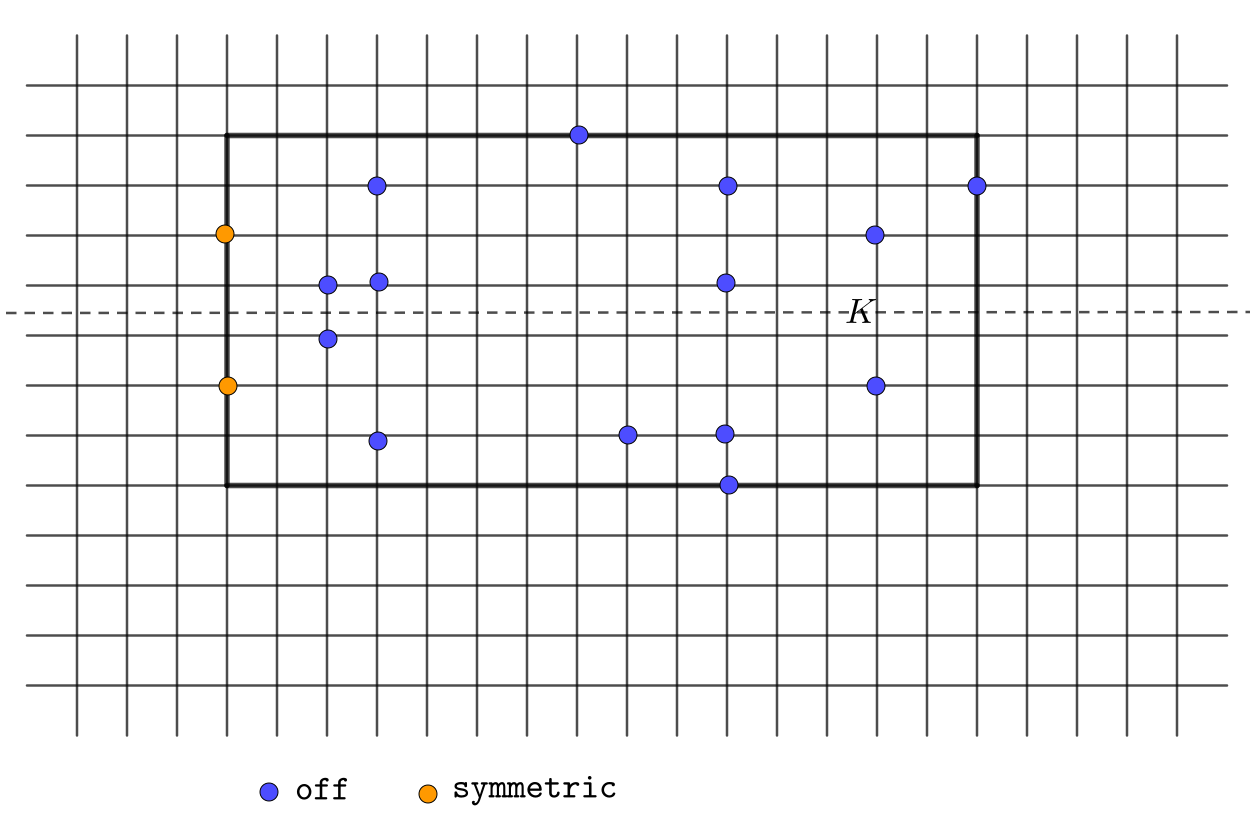}
     \caption{Robots with light \texttt{terminal1} change their lights to \texttt{symmetry}.}\label{Fig:Terminal1Symm2Symmetry}
    
   \end{center}
    
   \end{minipage}
\end{figure}

\begin{figure}[ht]\centering
   \begin{minipage}{0.45\textwidth}
    \includegraphics[width=1.15\linewidth]{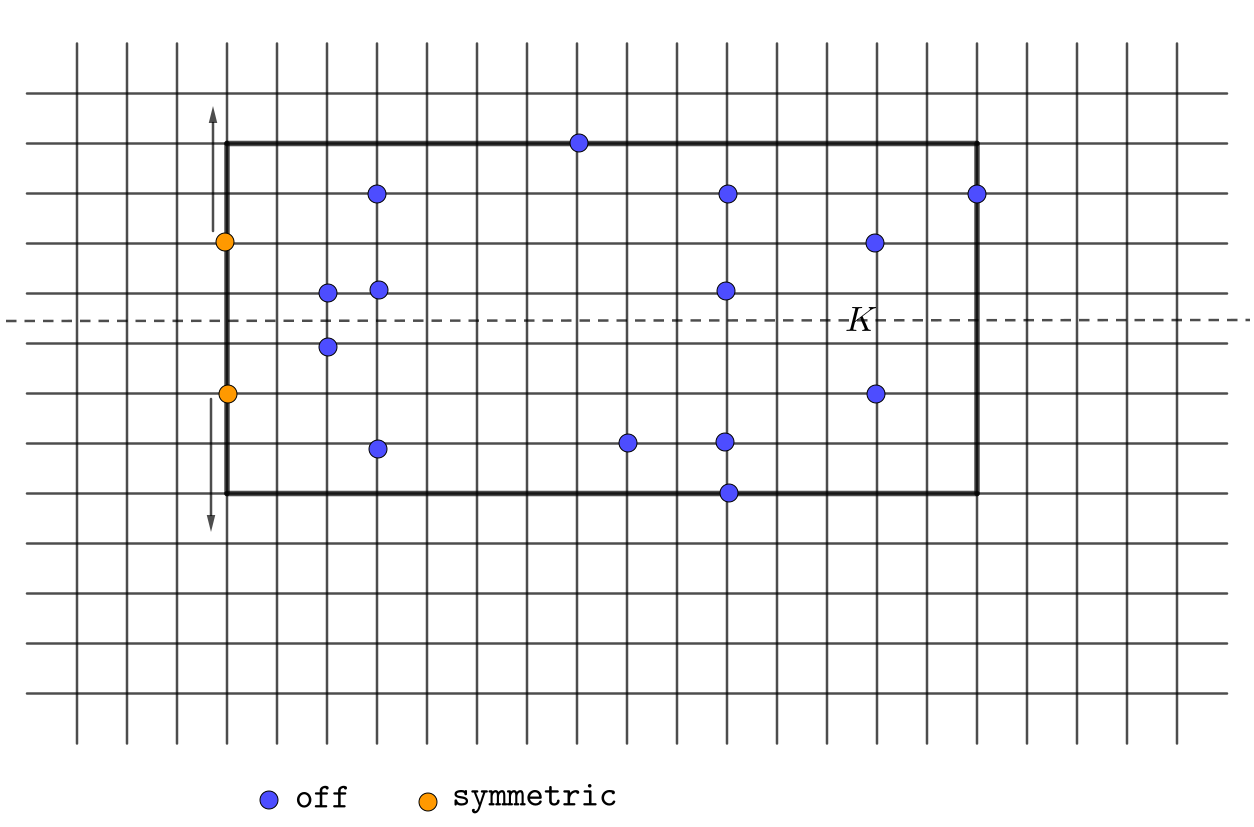}
     \caption{The robots with light \texttt{symmetry} move opposite of each other until they find either bottom or upper closed half has no other robots.}\label{Fig:SymmetryMovement}
   \end{minipage}
   \hspace{1mm}
   \begin{minipage}{0.45\textwidth}
    \includegraphics[width=1.15\linewidth]{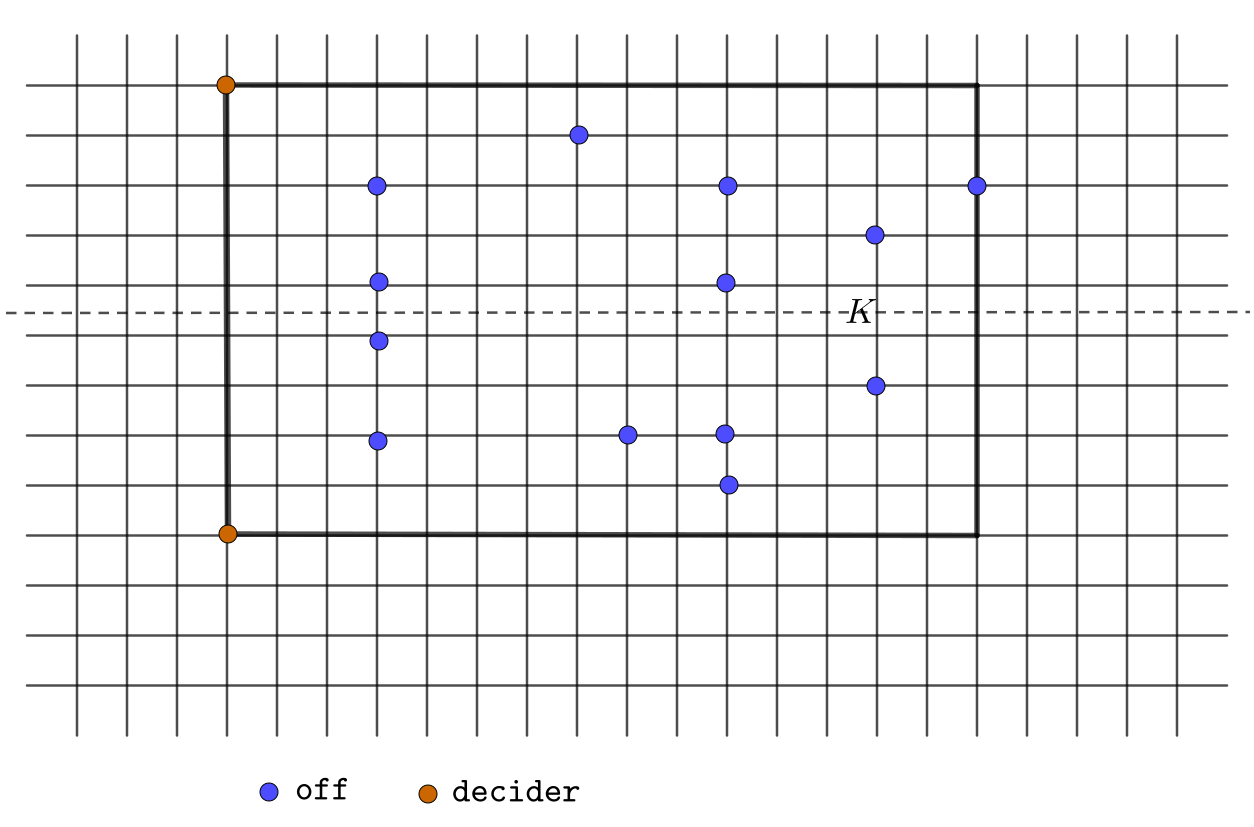}
     \caption{After reaching designated positions where upper closed or bottom closed half has no other robots, the robots with light \texttt{symmetry} change lights to \texttt{decider} to reach a stable configuration.}\label{Fig:StableConfig}
   \end{minipage}
\end{figure}

\begin{figure}[ht]

\end{figure}

  Otherwise if $\mathcal{R}_I(r_1)$ is symmetric with respect to $K$, then $r_1$ and $r_2$ change their light to \texttt{symmetric} (Figure \ref{Fig:Terminal1Symm}, \ref{Fig:Terminal1Symm2Symmetry}). Note that due to asynchrony, it may happen that one robot (say ,$r_1$) changes light to \texttt{symmetric} while $r_2$ still has the light \texttt{terminal1}. In this case, $r_1$ will not move until $r_2$ wakes and changes it's light to \texttt{symmetric}, after seeing the \texttt{symmetric} light of $r_1$. After both the robots $r_1$ and $r_2$ changed their lights to \texttt{symmetric}, they move vertically in the opposite direction of each other until one of the region $H_C^U$ or, $H_C^B$ has no other robots (Figure \ref{Fig:SymmetryMovement}). After this step, both $r_1$ and $r_2$ change their lights from \texttt{symmetric} to \textbf{\texttt{decider}}, thus reaching a stable configuration (Figure \ref{Fig:StableConfig}). Note that the robots will not check the symmetry of $\mathcal{R}_I(r_1) ( = \mathcal{R}_I(r_2)$) if the light of robots are \texttt{symmetric}.

  \textbf{Case 2:} Lets assume $r_1$  is on $\mathcal{L}_1$ and awake before another robot $r_2$ with \texttt{terminal1} light is yet to reach on $\mathcal{L}_1$. In this case, $r_1$ checks $\mathcal{R}_I(r_1)$ and finds out all the robots except one ($r_2$ with light \texttt{terminal1}) have light \texttt{off}. In this case, $r_1$ waits until $r_2$ reaches $\mathcal{L}_1$ and then by similar argument like in case 1, either reaches a configuration with a robot having light \texttt{leader1} or reaches a stable configuration.
  
  So after the execution of Algorithm \ref{Algo_Phase1}, the configuration is either stable or has exactly one robot with light \texttt{leader1}.
\end{proof}

\begin{lemma}
  \label{lemma:stage1_3}
  If the initial configuration $\mathbb{C}(0)$ has more than two robots on $\mathcal{L}_1$, then  $\exists$ $T^{''} >0$ such that $\mathbb{C}(T^{''})$ is either a stable configuration or there is exactly one robot with light \texttt{leader1} in  $\mathbb{C}(T^{''})$.
\end{lemma}
\begin{proof}
 Let us assume there are $h$ robots on the line $\mathcal{L}_1$, denoted by $\{r_i: i \in [1,h] \cap \mathbb{N}$ and $h >2\}$. Suppose $r_1$ and $r_2$ are the two terminal robots. Since $r_1$ and $r_2$ can identify themselves as terminal robots and can see there is no robot in \textbf{$H_{L}^{O}(r)$} and no  robot \texttt{leader1} in $\mathcal{R}_I(r)$, they will change their light to \texttt{terminal1} and move left. Note that after $r_1$ or $r_2$  move left, $\mathcal{L}_1$ now denotes the vertical line on which the robots are located now.
 
 Now, with the same argument as in  both the cases of Lemma \ref{lemma:stage1_2}, we can easily say that after the execution of Algorithm \ref{Algo_Phase1}, the configuration is either stable or has exactly one robot with light \texttt{leader1}.
\end{proof}
\begin{theorem}
For any initial configuration $\mathbb{C}(0)$, $\exists$ $T > 0$ such that either $\mathbb{C}(T)$ is a stable configuration or has exactly one robot with light \texttt{leader1}.
\end{theorem}
\begin{proof}
 This follows directly from Lemma \ref{lemma:stage1_1}, Lemma \ref{lemma:stage1_2} and Lemma \ref{lemma:stage1_3}. 
\end{proof}

\subsubsection{Phase 2}
After completion of \textit{Phase 1}, two configurations may occur. In the first possible configuration, there exists exactly one robot with light \texttt{leader1} and other robots with light \texttt{off}. The other possible configuration is a stable configuration. After completion of \textit{Phase 2}, these configurations transforms into a leader configuration. 

First, let us assume that after \textit{Phase 1}, the configuration transforms into the stable configuration. In this configuration, there are two robots with light \texttt{decider} such that their $H_U^C \cap H_L^C$ or $H_B^C \cap H_L^C$ have no other robots. All the other robots have light \texttt{off}. In \textit{Phase 2}, a robot with light \texttt{off} (say, $r$) checks whether it can see two robots with light \texttt{decider} on $\mathcal{L}_I(r)$. If $r$ can see two such robots on $\mathcal{L}_I(r)$ and $r$ is on $K \cap \mathcal{L}_V(r)$, then $r$ changes its light to \texttt{leader1}. Note that $r$ can calculate the line $K$ as it can see both the robots having light \texttt{decider}. On the other hand, if $r$ is not on $K \cap \mathcal{L}_V(r)$, then it checks the symmetry of $\mathcal{R}_I(r)$ with respect to $K$. If $\mathcal{R}_I(r)$ is not symmetric, then the terminal robot on $\mathcal{L}_V(r)$ on the dominant half changes its light to \texttt{leader1}. Otherwise, if $\mathcal{R}_I(r)$ is symmetric with respect to $K$, the robot closest to $K$ and on $\mathcal{L}_V(r)$ changes its light to \texttt{call}. Note that a robot with light \texttt{off} also  changes its light to \texttt{call} when it sees another robot with light \texttt{call} on the same vertical line. In this way, if $r$ is not on $K \cap \mathcal{L}_V(r)$ and $\mathcal{R}_I(r)$ is symmetric  with respect to $K$, all the robots on $\mathcal{L}_V(r)$ will eventually change their lights to \texttt{call} (Figure \ref{Fig:call1}, \ref{Fig:call2}).

Now let $r_{d1}$ be a robot with light \texttt{decider}. It checks whether there is a robot on $K \cap \mathcal{R}_I(r_{d1})$. Observe that in the beginning of \textit{Phase 2},  $r_{d1}$ can see  another robot with light \texttt{decider} (say $r_{d2}$) on $\mathcal{L}_V(r_{d1})$, thus can calculate the line $K$. If there is no robot on $\mathcal{R}_I(r_{d1}) \cap K$ and $\mathcal{R}_I(r_{d1})$ is not symmetric with respect to $K$, then if $r_{d1}$ was in dominant half, it changes its light to \texttt{leader1} (Figure \ref{Fig:Decider2Leader1}). In this case, since there are two robots with light \texttt{decider} and both are terminal, one will always be in the dominant half. Now if there is no robot on $\mathcal{R}_I(r_{d1}) \cap K$ and the line $\mathcal{R}_I(r)$ (where, $r \in \mathcal{R}_I(r_{d1})$) is symmetric with respect to $K$, after a certain time all the robots on $\mathcal{R}_I(r_{d1})$ will have light \texttt{call}. 
%In this situation $r_d$ will move right. Let the two decider robots $r_{d1}$ and $r_{d2}$ are on same vertical line and sees all robots have light \texttt{call} on $\mathcal{R}_I(r_{d1})$ ($=\mathcal{R}_I(r_{d2})$). 

  \begin{algorithm}[H]
     \setstretch{0.1}
    \SetKwInOut{Input}{Input}
    \SetKwInOut{Output}{Output}
    \SetKwProg{Fn}{Function}{}{}
    \SetKwProg{Pr}{Procedure}{}{}

    \Pr{\textsc{Phase2()}}{

    $r \leftarrow$ myself
    %\\$\mathcal{L}_i$ is a horizontal line with $i \geq 2$.
%     \\$\mathcal{L}_r$ = horizontal line where $r$ resides.\\$\mathcal{L}_1$ = lowest horizontal line.

    %\While{\textsc{LineFormation() = False}}{

    \uIf{$r.light =$ \texttt{decider}}
         {
         
         \uIf{there is a robot with light \texttt{leader1} on $\mathcal{R}_I(r)$ or $\mathcal{L}_V(r)$}{$r.light \leftarrow$ \texttt{off}}
         \Else{
              
              \uIf{there is a robot with light \texttt{decider} on $\mathcal{L}_{V}(r)$ and no robots on $K \cap \mathcal{R}_I(r)$}
                     {
                     \uIf{$\mathcal{R}_I(r)$ is symmetric with respect to $K$ }
                          {\If{all robots in $\mathcal{R}_I(r)$ are \texttt{call} }{move right}}
    
                    \Else{\If{$r$ is in dominant half}{$r.light \leftarrow$ \texttt{leader1}}  }
                     }
                     
               \uElseIf{there is a robot with light \texttt{decider} in $\mathcal{R}_I(r)$}{move right}
               
               %\uElseIf{there is a robot with light \texttt{leader} in $\mathcal{L}_V(r)$}{$r.light \leftarrow$ \texttt{off} \\move right}
               
               \ElseIf{there is a robot with light \texttt{call} in $\mathcal{L}_V(r)$ and no robot with light \texttt{decider} in $\mathcal{L}_I(r)$}{    
                           \If{all robots in $\mathcal{R}_I(r)$ are \texttt{call} }{move right }  
                             }

              }

         }
         
    \uElseIf{$r.light =$ \texttt{off}}
            {
            \If{there are two robots with light \texttt{decider} in $\mathcal{L}_I(r)$}
                      {
                      \uIf{$r$ is on $K \cap \mathcal{L}_V(r)$}{$r.light \leftarrow$ \texttt{leader1}}
                      \Else
                          {
                         \uIf{$\mathcal{R}_I(r)$ is symmetric with respect to $K$ }{
                                \If{$r$ is closest to $K$ or there          is a robot with light    \texttt{call} on $\mathcal{L}_V(r)$}{$r.light \leftarrow$ \texttt{call}}
                                 }
                         \Else{
                         
                              \If {$r$ is terminal on $\mathcal{L}_V(r)$ and in dominant half}{$r.light \leftarrow$ \texttt{leader1}}
                              }
                                  
                           }
                      
                       }
            
             }

    \uElseIf{$r.light =$ \texttt{call}}
            {
            
            \If{there is a robot with light \texttt{leader1} in $\mathcal{R}_I(r)$ or $\mathcal{L}_I(r)$}{$r.light \leftarrow$ \texttt{off}}
            
            }

    \ElseIf{$r.light =$ \texttt{leader1}}
            {
        \uIf{(all robots in $\mathcal{L}_I(r)$ are \texttt{off}) or ($H_{L}^{O}(r)$ is empty and all robots in $\mathcal{R}_I(r)$ are \texttt{off})}    
            {
            \uIf{there is no other robot in $H_{U}^{C}(r)$ or        $H_{B}^{C}(r)$}{
                            \uIf{there is other robot in $H_{L}^{C}(r)$      }{move left}
                            \Else{
                            $r.light \leftarrow$ \texttt{leader}}
                            
                            }
                            
            \Else{

                         \uIf{$r$ is not terminal on $\mathcal{L}_V(r)$}{move left}
                              
                         \uElseIf{$r$ is terminal on $\mathcal{L}_V(r)$ and there is another robot $r'$ on $\mathcal{L}_V(r)$}{move vertically opposite to $r'$}
                         
                         \ElseIf{$r$ is singleton on $\mathcal{L}_V(r)$}{move vertically according to its positive $y-axis$.}
                 }        
             }            
                         
         \ElseIf{
                  there is a robot with light \texttt{leader1} in $\mathcal{L}_I(r)$}
                  {$r.light \leftarrow$ \texttt{off}
                         }

            }

  }

    \caption{\textbf{Phase 2}}
    \label{leader selection}
\end{algorithm}

 \newpage

\begin{figure}[!htb]\centering
   \begin{minipage}{0.48\textwidth}
     \includegraphics[width=1.15\linewidth]{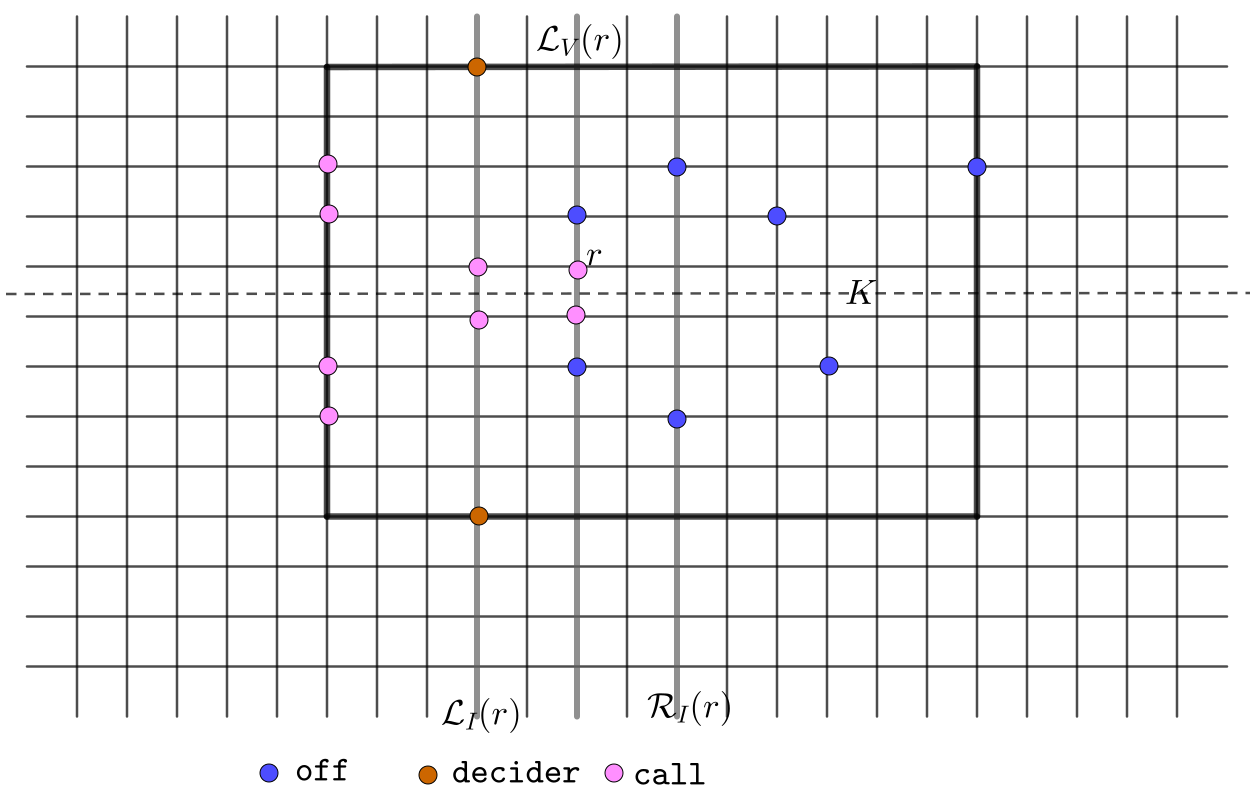}
     \caption{Robot $r$ (closest of $K$) sees two \texttt{decider} robot on $\mathcal{L}_I(r)$ and $\mathcal{R}_I(r)$ is symmetric, so it changes light to \texttt{call}. }\label{Fig:call1}
   \end{minipage}
   \hspace{1.5mm}
   \begin{minipage}{0.45\textwidth}
    \includegraphics[width=1.15\linewidth]{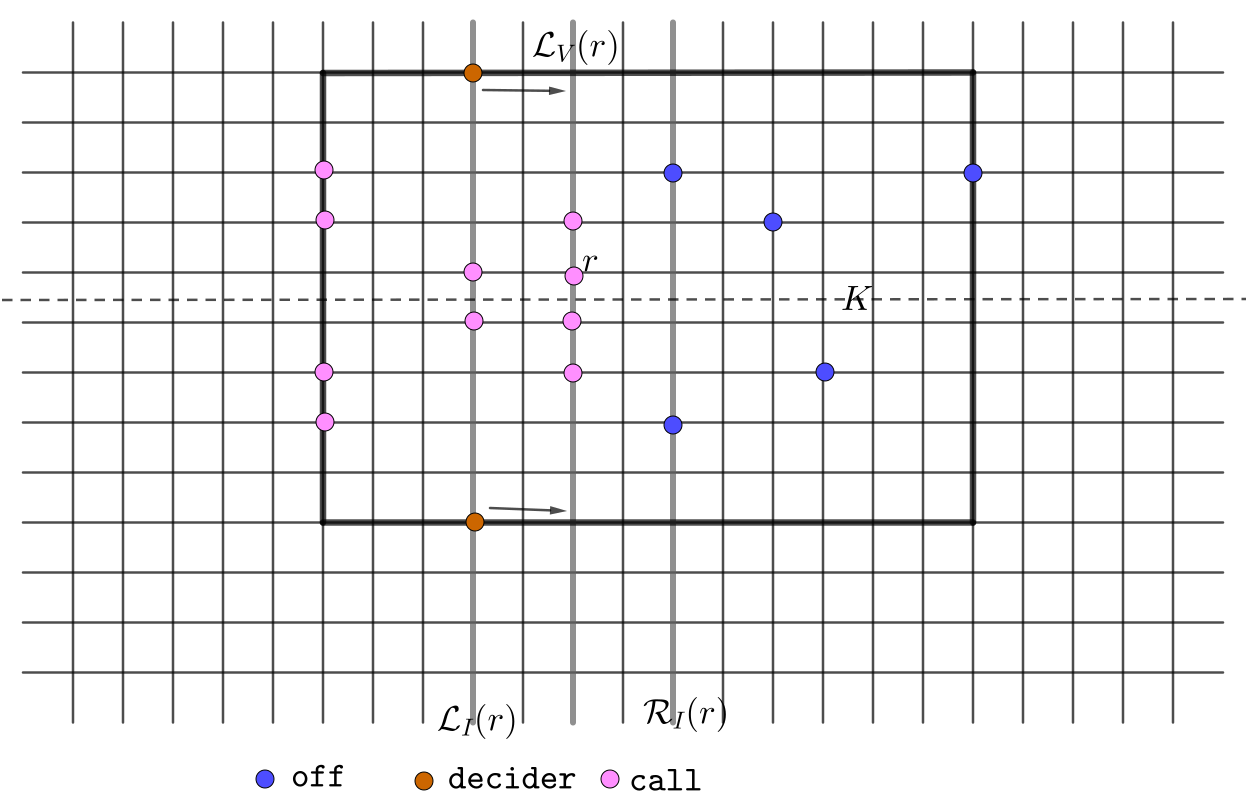}
     \caption{The robots on $\mathcal{L}_V(r)$, when see a robot with light \texttt{call} on the same vertical line, change their light to \texttt{call}. After all robots on $\mathcal{L}_V(r)$ have \texttt{call} light, the \texttt{decider} robots move right.}\label{Fig:call2}
   \end{minipage}
\end{figure}

Then if $\mathcal{R}_I(r_{d1})$ ($=\mathcal{R}_I(r_{d2})$) is symmetric, both the robots $r_{d1}$ and $r_{d2}$ move right. In case, due to asynchrony if one robot (say, $r_{d1}$) moves right and wakes before $r_{d2}$ moves, then even if $r_{d1}$ sees all robots on $\mathcal{R}_I(r_{d1})$ with light \texttt{call}, it will not move right until it sees there is no robot with light \texttt{decider} on $\mathcal{L}_I(r_{d1})$ and sees a robot with light \texttt{call} or $r_{d2}$ on $\mathcal{L}_V(r_{d1})$. Note that $r_{d1}$ will see $r_{d2}$ on $\mathcal{L}_I(r_{d1})$ until $r_{d2}$ moves right. On the other hand, $r_{d2}$ sees $r_{d1}$ on $\mathcal{R}_I(r_{d2})$ and then it moves right.
 
On a lighter note, in this algorithm a robot $r$ with light \texttt{off} basically calculates the line $K$ by seeing the robots with light \texttt{decider} on $\mathcal{L}_I(r)$ and checks the symmetry of $\mathcal{R}_I(r)$ with respect to the line $K$ and checks if it is on $K$. Note that, the robots with light \texttt{decider} also can calculate $K$ when they see each other on the same vertical line and check the symmetry of the next vertical line having robots and also do search for robots on $K$. As the configuration is solvable and movement of robots with light \texttt{decider} do not actually make the configuration unsolvable, after a certain time there will be some robot $r_0$ with light \texttt{leader1}. Note that a robot $r$ with light \texttt{call} changes the light to \texttt{off} if it sees any robot  with light \texttt{leader1} on $\mathcal{R}_I(r)$ or $\mathcal{L}_I(r)$. 

Now, let us consider the case where there is a robot with light \texttt{leader1} (say, $r_0$). It first checks whether all robots on the line $\mathcal{L}_I(r_0)$ have light \texttt{off} or if its left open half is empty. If one of the cases becomes true and all robots on $\mathcal{R}_I(r_0)$ have light \texttt{off}, then $r_0$ checks whether one of $H_U^C(r_0)$ or $H_B^C(r_0)$) has no other robots. If
so and also there is other robots on the left closed half ($H_L^C(r_0)$) of $r_0$, then it moves left until $H_L^C(r_0)$ has no other robots. Then it changes its light to \texttt{leader}. On the other hand, if both $H_U^C(r_0)$ and $H_B^C(r_0)$ have other robots, then $r_0$ checks whether it is a terminal robot on $\mathcal{L}_V(r_0)$. Let us assume $r_0$ is terminal on $\mathcal{L}_V(r_0)$. Now, if $r_0$ is singleton on $\mathcal{L}_V(r_0)$, then it moves according to its positive $y$-axis until either $H_U^C(r_0)$ or $H_B^C(r_0)$ has no other robots. If $r_0$ is not singleton but terminal on $\mathcal{L}_V(r_0)$, then there exists a robot (say, $r'$) on $\mathcal{L}_V(r_0)$. In this case, $r_0$ moves vertically in the opposite direction of $r'$ until $H_U^C(r_0)$ or $H_B^C(r_0)$ has no other robots. If $r_0$ was not terminal on $\mathcal{L}_V(r_0)$, then it moves left. Note that using such movements, $r_0$ always reaches a grid point such that either $H_U^C(r_0) \cap H_L^C(r_0)$ or $H_B^C(r_0) \cap H_L^C(r_0)$ has no other robots where it changes its light to \texttt{leader}. Also, it might happen that a robot with light \texttt{leader1} already sees another robot with light \texttt{leader1} (Figure \ref{Fig:2Leader1_1}). In this case, the robot with light \texttt{leader1}, who sees the other robot with light \texttt{leader1} on its left, changes light to \texttt{off} (Figure \ref{Fig:2Leader1to1Leader1}). Also if a robot with light \texttt{decider} (say, $r$) finds a robot with light \texttt{leader1} on $\mathcal{R}_I(r)$, then it changes its light to \texttt{off}. Thus after a certain time, there is exactly one robot  with light \texttt{leader} and others have light \texttt{off}.

So, after completion of \textit{Phase 2}, the configuration transforms into a leader configuration. The following Lemmas \ref{p2:l4}, \ref{P2:l5}, \ref{P2:l6} and \ref{P2:l7} justify the correctness of the Algorithm \ref{leader selection}.

\begin{lemma}
\label{p2:l4}
In \textit{Phase 2}, a robot $r$ with light \texttt{decider} changes its light to \texttt{leader1} only if $H_L^C(r)$ has no other robots.
\end{lemma}
\begin{proof}
 Let $r$ and $r'$ be two robots with light \texttt{decider} at some time $T>0$ in \textit{Phase 2} on the same vertical line. Without loss of generality, let us assume that  $r$ be the robot which changes its light at  time $T_1 > T$ to \texttt{leader1} while it was still on the same vertical line at time $T$. Note that by Algorithm \ref{leader selection}, this only happens if $\mathcal{R}_I(r)$ is not symmetric with respect to $K$ and $r$ is on the dominant half at $T_1$. We will denote the vertical line on which a robot $r$ situated at any time $t$ by $L_V^{t}(r)$.
 
 Let us assume that $H_L^C(r)$ has other robots.
%  Now let us assume again that $r$ and $r'$ both can not see each other on $L_V^T(r)$. This implies there are other robots on $L_V^T(r)$, between $r$ and $r'$.
 Now note that a robot $r$ with light \texttt{decider} moves right if it sees all the robots on $\mathcal{R}_I(r)$ have light \texttt{call} or sees another robot with light \texttt{decider} on $\mathcal{R}_I(r)$. Also in the beginning of \textit{Phase 2}, $H_L^C(r)$ ($= H_L^C(r')$) has no other robots. So, at the time $T$, $H_L^C(r)$ has other robots and still has light \texttt{decider}, implies $r$ and $r'$ moved right at least once and did not see any robot with light \texttt{leader1} on $\mathcal{R}_I(r)$ till that time. Let us assume that at a time $(0 <) T_2< T$, $r$ and $r'$ were on $L_V^{T_2}(r) $ which is actually $L_I(r)$ at time $T$. Now $r$ and $r'$ moves to $L_V^T(r)$ from $L_V^{T_2}(r)$ only when all robots on $L_V^T(r)$ have light  \texttt{call}. That is only possible if at time $T$, the robots on $R_I(r)$ are symmetric with respect to $K$. Thus it is impossible for $r$ to change its light to \texttt{leader1} while still on the same line $L_V^T(r)$ at time $T_1$. So our assumption that $H_L^C(r)$ has other robots at time $T$, is wrong. Hence the result follows (Figure \ref{Fig:Decider2Leader1}).
 
 \begin{figure}[ht]
     \includegraphics[width=0.6\linewidth]{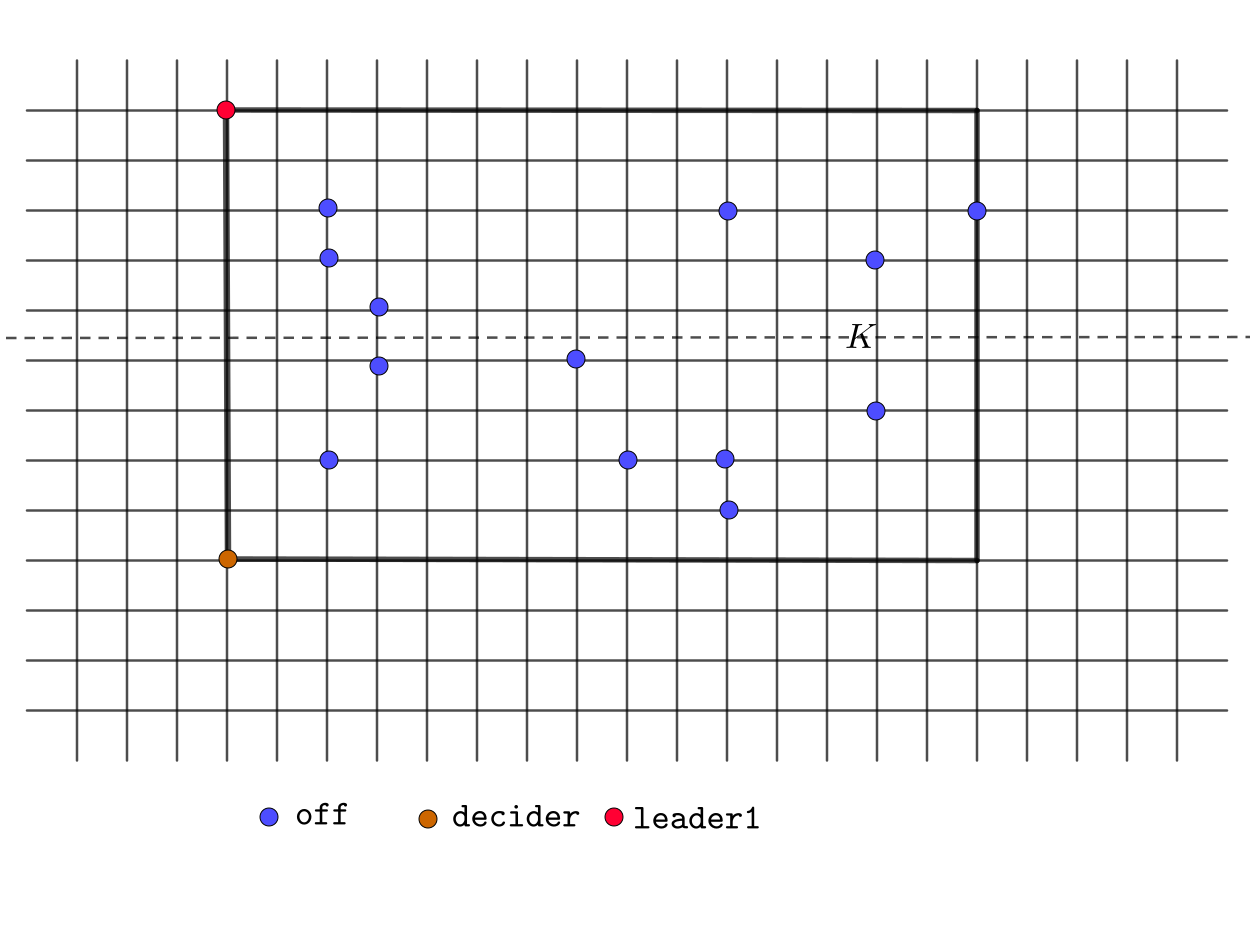}
     \caption{The robot $r$ with light \texttt{decider} changes its light to \texttt{leader1} with $H_L^C(r)$ has no other robots.}\label{Fig:Decider2Leader1}
\end{figure}
\end{proof}

\begin{lemma}
\label{P2:l5}
If after a certain time $T>0$ in \textit{Phase 2}, the configuration $\mathbb{C}(T)$ has two robots $r$ and $r'$ both with light \texttt{leader1}, then $\exists$ $0<T_1,T_2<T$ such that $r$ (or $r'$) changed its light to \texttt{leader1} from \texttt{decider} at $T_1$  and $r'$ (or, $r$) changed its light to \texttt{leader1} from \texttt{off} at $T_2$.
\end{lemma}
\begin{proof}
 First, we will show that both the robots $r$ and $r'$ can not have changed their lights to \texttt{leader1} from light \texttt{decider} at $T_1$ and $T_2$. From the Algorithm \ref{leader selection}, it is clear that a robot $r$ with light \texttt{decider} changes its light to \texttt{leader1} only if $\mathcal{R}_I(r)$ is not symmetric, $K \cap \mathcal{R}_I(r)$ is empty and $r$ is in dominant half. Now both the robots with light \texttt{decider} can not be on the dominant half. So, both $r$ and $r'$ can not have change their lights to \texttt{leader1} from \texttt{decider} at $T_1$ and $T_2$.
 
 Secondly, it has to be shown that both $r$ and $r'$ can not have changed their lights to \texttt{leader1} from light \texttt{off} at $T_1$ and $T_2$. From Algorithm \ref{leader selection}, it is obvious that a robot $r_0$ with light \texttt{off} can only change its light to \texttt{leader1} if it sees two decider robots $r_{d1}$ and $r_{d2}$ on $\mathcal{L}_I(r_0)$. So if at a time $T'< T_1, T_2$ both $r$ and $r'$ had light \texttt{off}, then $r$ and $r'$ was on the same vertical line $\mathcal{L}_V(r)$ ($=\mathcal{L}_V(r')$). Now we will show by contradiction that it is not possible for both $r$ and $r'$ to change their lights to \texttt{leader1}. Let us assume $\mathcal{R}_I(r)$ ($=R_I(r')$) is symmetric with respect to $K$. Then $r$ or $r'$ will only change their lights to \texttt{leader1} if both are on the grid position $K \cap \mathcal{R}_I(r_{d1})$. Which is not possible. So, let us now assume $\mathcal{R}_I(r)$ is not symmetric with respect to $K$ and also let one of $r$ or, $r'$ is on $K \cap \mathcal{R}_I(r_{d1})$. Without loss of generality, let $r$  is on $K \cap \mathcal{R}_I(r_{d1})$. Then according to the Algorithm \ref{leader selection}, $r$ changes its light to \texttt{leader1} but no other robot on $\mathcal{L}_V(r)$ changes their lights as they are not closest to $K$ or sees no robot with light \texttt{call} on $\mathcal{L}_V(r)$. Hence in this case, there can be only one robot on $\mathcal{L}_V(r)$ who will change its light to \texttt{leader1}, arriving at a contradiction again. So let $\mathcal{R}_I(r)$ is not symmetric and there is no robot on the grid point $K \cap \mathcal{R}_I(r_{d1})$. Then only the terminal robot of $\mathcal{L}_V(r)$ who is in the dominant half changes its color to \texttt{leader1}. So again a contradiction that both $r$ and $r'$ on $\mathcal{L}_V(r)$ change their light to \texttt{leader1}. Hence our assumption was wrong.
 
 \begin{figure}[!htb]\centering
   \begin{minipage}{0.45\textwidth}
     \includegraphics[width=1.15\linewidth]{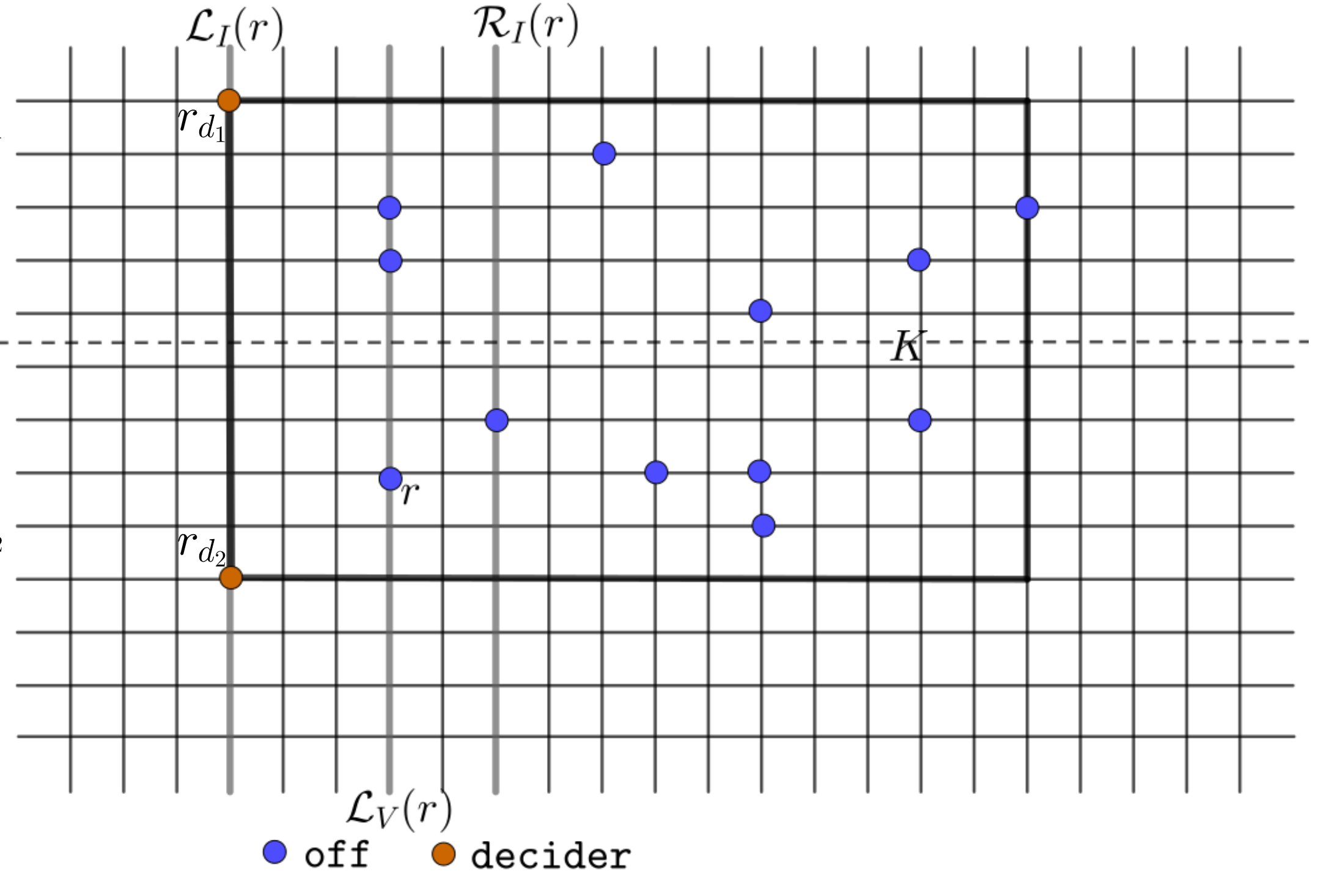}
     \caption{The robot $r_{d1}$ with light \texttt{decider} and $r$ with light \texttt{off} both are on the dominant half.}\label{Fig:terminal2_2Leader1}
   \end{minipage}
   \hspace{1.5mm}
   \begin{minipage}{0.45\textwidth}
    \includegraphics[width=1.15\linewidth]{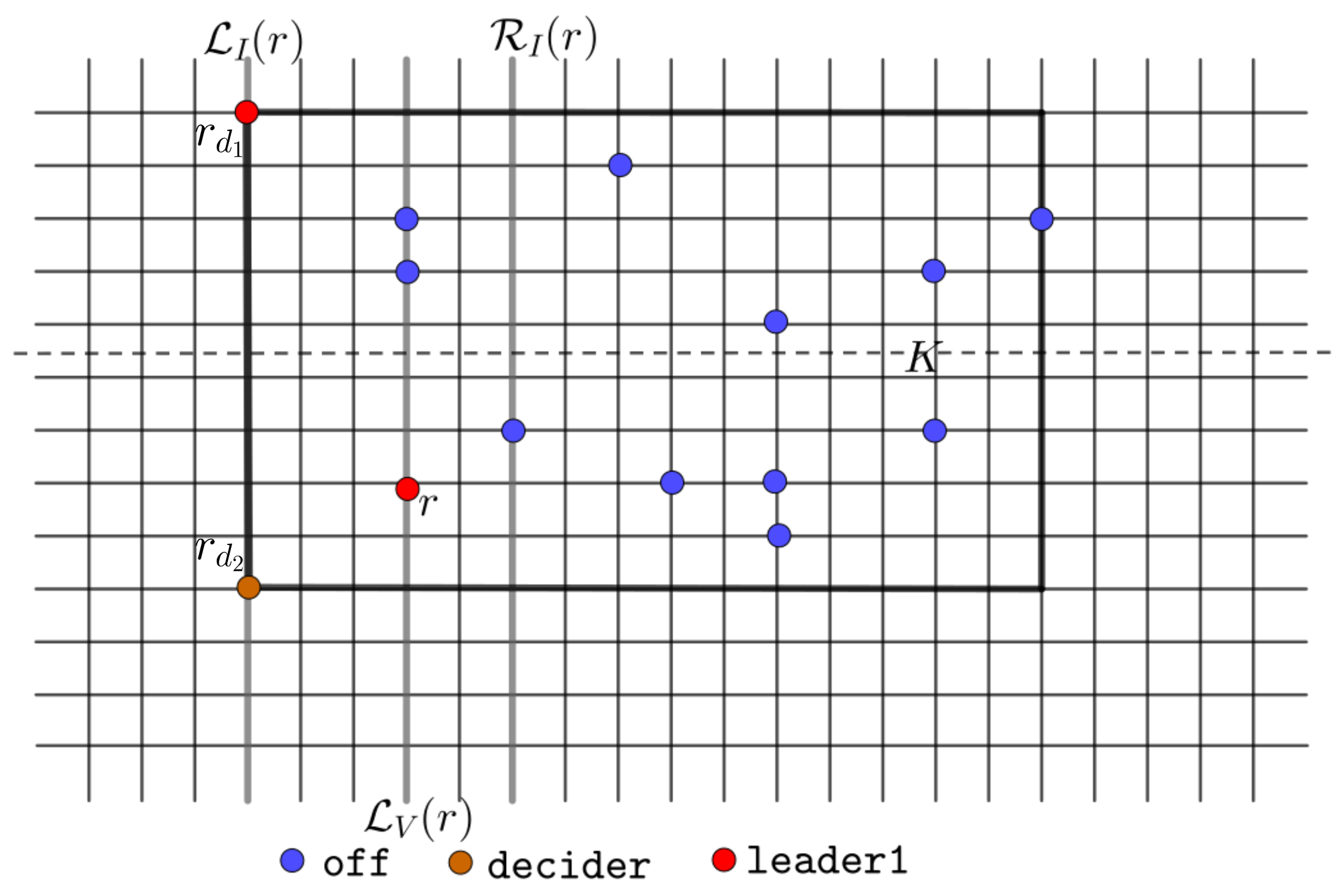}
     \caption{The robot $r_{d1}$ with light \texttt{decider} and $r$ with light \texttt{off} change their lights to \texttt{leader1}.}\label{Fig:2Leader1}
   \end{minipage}
\end{figure}

 Hence one of $r$ or $r'$ will change its light to \texttt{leader1} from light \texttt{decider} at a time $T_1 < T$ and the other robot will change its light to \texttt{leader1} from light \texttt{off} at time $T_2 < T$ (Figure \ref{Fig:terminal2_2Leader1}, \ref{Fig:2Leader1}).

\end{proof}

\begin{lemma}
\label{P2:l6}
 If in \textit{Phase 2}, there exists a configuration $\mathbb{C}(T)$ at a time $T>0$ such that $\exists$ a robot $r'$ with light \texttt{leader1} which sees another robot $r$ with light \texttt{leader1} on $L_I(r')$ in $\mathbb{C}$, then $r'$ will not move and change its light to \texttt{off}.
\end{lemma}
\begin{proof}
Since at time $T$, there are two robots with light \texttt{leader1} in the configuration during \textit{Phase 2}, robot $r$ must have changed its light to \texttt{leader1} from light \texttt{decider} and the other robot $r'$ must have changed the light to \texttt{leader1} from \texttt{off} (by Lemma \ref{P2:l5}). Now from Lemma \ref{p2:l4}, $H_L^C(r)$ has no other robots. Also $r'$ is on $\mathcal{R}_I(r)$ (as $r'$ can only change its light to \texttt{leader1} if it has seen two robots with light \texttt{decider} on $L_I(r')$). Thus it can be seen that if the configuration has two robots with light \texttt{leader1} at any time $T$, the two robots will be on two consecutive occupied vertical line. %Also note that the robot who has changed its light to \texttt{leader1} from light \texttt{off} (i.e $r'$) is on $R_I(r)$. 
So by Algorithm \ref{leader selection}, $r'$ will see $r$ on $L_I(r')$ with light \texttt{leader1} and change its light to \texttt{off} without moving (Figure \ref{Fig:2Leader1_1}, \ref{Fig:2Leader1to1Leader1}).

\begin{figure}[!htb]\centering
   \begin{minipage}{0.45\textwidth}
     \includegraphics[width=1.15\linewidth]{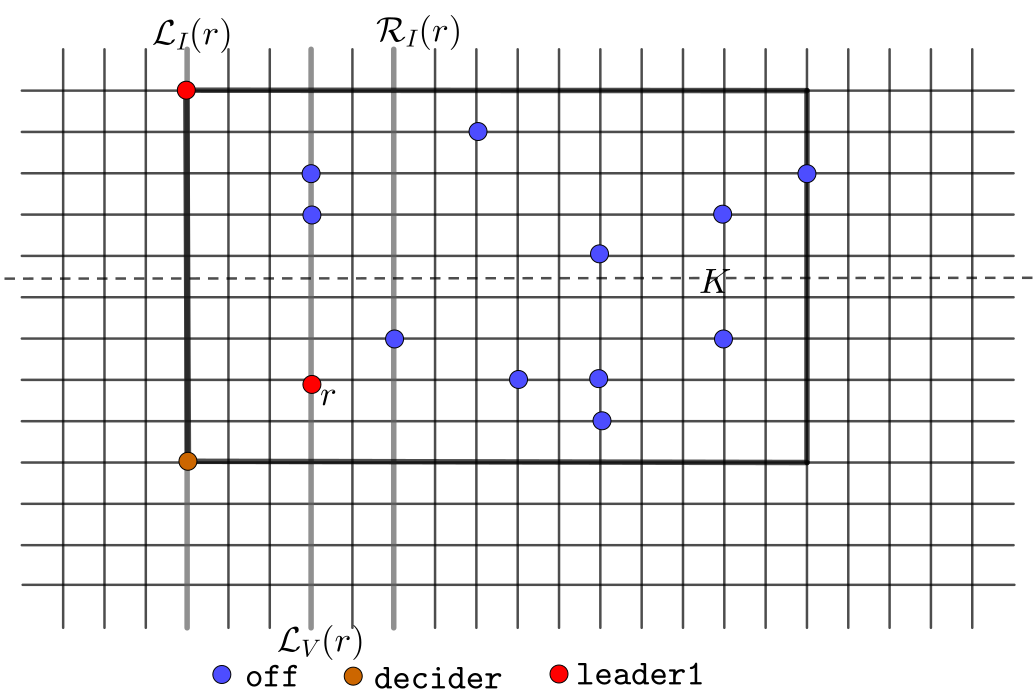}
     \caption{There are two robots with light \texttt{leader1} in the configuration.}\label{Fig:2Leader1_1}
   \end{minipage}
   \hspace{1.5mm}
   \begin{minipage}{0.45\textwidth}
    \includegraphics[width=1.15\linewidth]{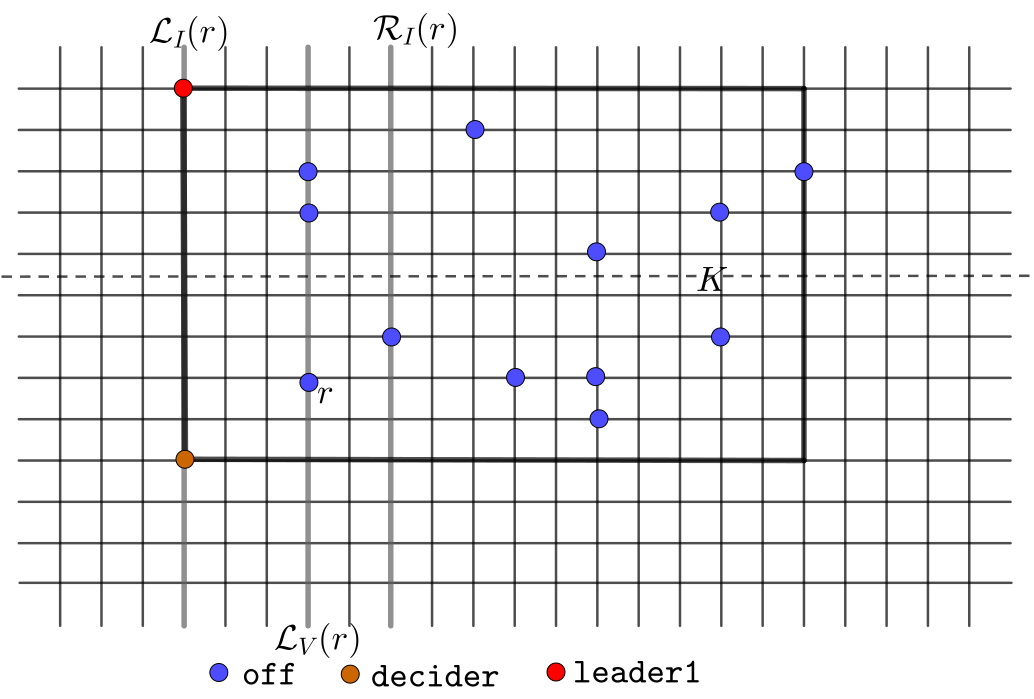}
     \caption{The robot $r$ changes its light to \texttt{off} after seeing another \texttt{leader1} robot on $L_I(r)$.}\label{Fig:2Leader1to1Leader1}
   \end{minipage}
\end{figure}
\end{proof}

\begin{lemma}
\label{P2:l7}
If a robot $r$ with light \texttt{leader1} is not terminal on $\mathcal{L}_V(r)$ and does not see any robot with light \texttt{leader1} on $\mathcal{L}_I(r)$, then $H_L^O(r) \cap \mathcal{L}_H(r)$ does not contain any robot.
\end{lemma}
\begin{proof}
 Let $r$ is not terminal on $\mathcal{L}_V(r)$ with light \texttt{leader1}. Then by Algorithm \ref{leader selection}, $r$ is on $K \cap R_I(r_{d_1})$ $(\mathcal{R}_I(r_{d1})=\mathcal{R}_I(r_{d2}))$, where $r_{d_1}$ and $r_{d_2}$ be two robots with light \texttt{decider}. If possible let, there is a robot $r'$ on  $H_L^O(r) \cap \mathcal{L}_H(r)$. Then $\exists$ $T > 0$ such that $r_{d_1}$ and $r_{d_2}$ are on $L_I(r')$. Now observe that at time $T$, $r'$ is on $K \cap R_I(r_{d_1})$, which implies $r'$ will change its light to \texttt{leader1}.
 %Also observe that at $T$, $r_{d_1}$ and $r_{d_2}$ are not on %$\mathcal{L}_I(r)$. Now two cases can happen.
%In the first case either one of $r_{d_1}$ or $r_{d_2}$ changes its light to \texttt{leader1} also. In that case $r'$ changes its light to \texttt{off}. observe that in this case $r_{d_1}$ and $r_{d_2}$ will never reach $\mathcal{L}_I(r)$ thus $r$ can never change its light to \texttt{leader1}.
%Or,in the second case both
Then both $r_{d_1}$ and $r_{d_2}$ change their lights to \texttt{off} seeing $r'$ on $R_I(r_{d_1})$ with light \texttt{leader1}. Hence $r_{d_1}$ and $r_{d_2}$ never reach $\mathcal{L}_I(r)$. Thus $r$ can never change its light to \texttt{leader1}.

\begin{figure}[ht]
     \includegraphics[width=0.6\linewidth]{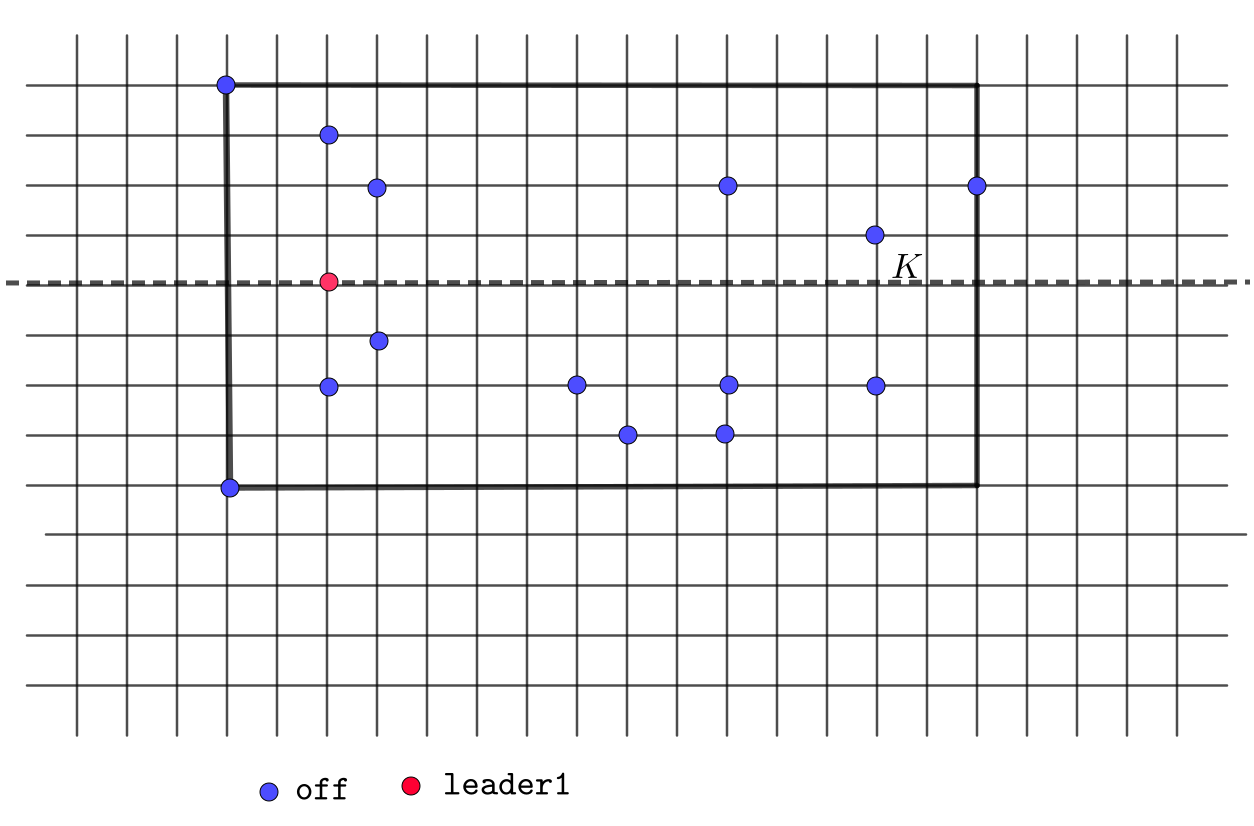}
     \caption{The robot with light \texttt{leader1} has its horizontal line empty on the left.}\label{Fig:empty left}
\end{figure}

%So, in both the cases we are arriving at contradictions.
So we can conclude that our primary assumption was wrong. Hence the result follows (Figure \ref{Fig:empty left}).
\end{proof}

\subsection{Pattern formation from leader configuration}
\label{ss2}
In this phase, initially the configuration is a leader configuration.  Note that the robots can agree on a global co-ordinate system based on the position of the robot $r_0$ with light \texttt{leader}. We denote the position of $r_0$ with the co-ordinate $(0,-1)$. Also all robots with light \texttt{off} lie on one of the open half planes delimited by the horizontal line $\mathcal{L}_H(r_0)$. This half plane will correspond to positive direction of $Y$-axis (Figure \ref{Fig:leaderConfig}). 

  \begin{algorithm}[H]
    \setstretch{.1}
    \SetKwInOut{Input}{Input}
    \SetKwInOut{Output}{Output}
    \SetKwProg{Fn}{Function}{}{}
    \SetKwProg{Pr}{Procedure}{}{}

  \Input{The configuration of robots visible to me.}
    
    \Pr{\textsc{PatternFormationFromLeaderConfiguration()}}{

    $r \leftarrow$ myself
    
    $r_0 \leftarrow$ the robot with light \texttt{leader}
    
    \uIf{$r.light =$ \texttt{off}}{
    
      \uIf{$($$r_0 \in {H}_B^O(r)$$)$  and $($$r$ is leftmost on $\mathcal{L}_H(r)$$)$ and $($there is no robot in ${H}_{B}^{O}(r) \cap {H}_{U}^{O}(r_0)$$)$ \label{code s2: 0}}{
    
	  \uIf{there are no robots on $\mathcal{L}_H(r_0)$ other than $r_0$ \label{code s2: 1}}{
	  
	    \uIf{there is a robot with light \texttt{done}\label{code s2: 2}}{
	    
	    \uIf{$r$ is at $t_{n-2}$}{$r.light \leftarrow$ \texttt{done}}
	    \Else{\textsc{TargetMove}$(n-2)$}
	    
	    }
	    \Else{\textsc{LineMove}$(1)$}
	  
	  }
	  \uElseIf{there are $i$ robots on $\mathcal{L}_H(r_0)$ other than $r_0$ at $(1,-1), \ldots, (i,-1)$\label{code s2: 3}}{\textsc{LineMove}$(i+1)$}
	  \ElseIf{there are $i$ robots on $\mathcal{L}_H(r_0)$ other than $r_0$ at $(n-i,-1), \ldots, (n-1,-1)$\label{code s2: 6}}{
	  
	    \uIf{$r$ is at $t_{n-i-2}$}{$r.light \leftarrow$ \texttt{done}}
	    \Else{\textsc{TargetMove}$(n-i-2)$}
	  
	  }
    
      }
      
      \ElseIf{$r_0 \in \mathcal{L}_H(r)$ and ${H}_U^O(r)$ has no robots with light \texttt{off}\label{code s2: 4}}{
      
	\If{$r$ is at $(i,-1)$}{Move to $(i,0)$}
      
      }

    }

    \ElseIf{$r.light =$ \texttt{leader}}{
    
      \If{there are no robots with light \texttt{off} \label{code s2: 5}}{
    
    \uIf{$r$ is at $t_{n-1}$}{$r.light \leftarrow$ \texttt{done}}
	    \Else{\textsc{TargetMove}$(n-1)$}}
    
      }
    }
  \Pr{\textsc{LineMove}$(j)$}{
  \uIf{$r$ is on $L_{H1}$}
             {\uIf{$r$ is at $(j,0)$}{Move to $(j,-1)$}
              \Else{Move horizontally towards $(j,0)$}
             }
  \Else{Move vertically towards $L_{H1}$}
  
  }
  
  \Pr{\textsc{TargetMove}$(j)$}{
  \uIf{$r$ is on $L_{t_j-1}$}
             {\uIf{$r$ is at $(t_j(x),t_j(y)-1)$}{Move to $(t_j(x),t_j(y))$}
              \Else{Move horizontally towards $(t_j(x),t_j(y)-1)$}
             }
  \Else{Move vertically towards $L_{t_j-1}$}
  
  }

\caption{Pattern Formation from Leader Configuration}
    \label{algo_phase3} 
\end{algorithm}

Therefore an agreement on a global co ordinate system can happen between the robots who see $r_0$ at $(0,-1)$. After  completion of this phase, the robots achieve the target configuration (Figure \ref{Fig:TargetEmbedding}). The robots first form a compact line and from that line the robots then move to their designated target positions. The difficulty of this phase is to differentiate between two configurations where a robot is going to form a compact line and where the robot is going to its target position which are described in  subsections 3.2.1 and 3.2.2.
\subsubsection{Compact line formation} 
Observe that according to the Algorithm \ref{algo_phase3}, a robot with light \texttt{leader} will not move during the formation of line. Also at the beginning of \textit{Phase 3}, there are no other robots on the line $\mathcal{L}_H(r_0)$. A robot $r$ with light \texttt{off} will first check if it can see $r_0$ and if it is the leftmost robot on the line $\mathcal{L}_H(r)$ and also if there are any robots on $H_B^O(r) \cap H_O^U(r_0)$. If all the conditions are true, (i.e $r$ is the leftmost robot in its horizontal line and there are no other robots in between horizontal lines of $r$ and $r_0$) $r$ counts the number of robots on $\mathcal{L}_H(r_0)$. Lets assume if there are no robots on $\mathcal{L}_H(r_0)$ except $r_0$. In this case, $r$ checks for other robots with light \texttt{done} on the grid. Note that a robot changes its light to \texttt{done} only if it has reached its target position. So, clearly while forming the line, $r$ finds that there are no robots with light \texttt{done} on the grid and  moves to the position $(1,-1)$ following the procedure \textsc{Linemove(1)}. On the other hand if there are  i robots except $r_0$ who are already in the line   $\mathcal{L}_H(r_0)$ occupying the positions $(1,-1), (2,-1)...(i,-1)$,  then $r$ simply move to $(i+1,-1)$ following the procedure \textsc{Linemove($i+1$)}.

  During the procedure \textsc{Linemove($j$)}, a robot (say, $r$) can recognize if it is on the line $L_{H1}$, where $L_{H1}$ is the immediate horizontal line above $\mathcal{L}_H(r_0)$. If it is not on $L_{H1}$, it moves vertically downwards until it reaches $L_{H1}$. Otherwise, if $r$ is already on $L_{H1}$, it checks if it is on the co-ordinate $(j,0)$. If it is not on $(j,0)$, it moves horizontally to $(j,0)$. Observe that during this horizontal movement, there will be no collision as $r$ will be the only robot on $L_{H1}$ due to the fact that the robots in Algorithm \ref{algo_phase3} move sequentially. Now, if $r$ is on the co-ordinate $(j,0)$, it moves vertically to the co-ordinate $(j,-1)$ (Figure \ref{Fig:LineFormation}, \ref{Fig:LineFormed}).
  Note that since the robots move sequentially, no robot will change there light to \texttt{done} before forming the compact line.
\subsubsection{Target Pattern Formation}
After formation of the compact line, the robots now will move to the target co-ordinates. Note that the target co-ordinates are unique for each robot who can see the robot $r_0$ with light \texttt{leader} as there is a agreement on global co-ordinate. The target co-ordinates are denoted as $t_i$, where $i \in \{0,1,2,...\overline{n-1}\}$. Also observe that if $t_i$ and $t_j$ are in the same horizontal line where $i < j $, then $t_i$ is on the right of $t_j$. And if $t_i$ and $t_j$ are not on the same horizontal line, then $t_i$ will be above of $t_j$. 

If a robot sees  $r_0$ on the same horizontal line, it vertically moves to $L_{H1}$ with y-coordinate 0. Now observe that after moving to $L_{H1}$ with y-coordinate 0, it can see all the robots on the line including the robot $r_0$. 

\begin{figure}[!htb]\centering
   \begin{minipage}{0.45\textwidth}
     \includegraphics[width=1\linewidth]{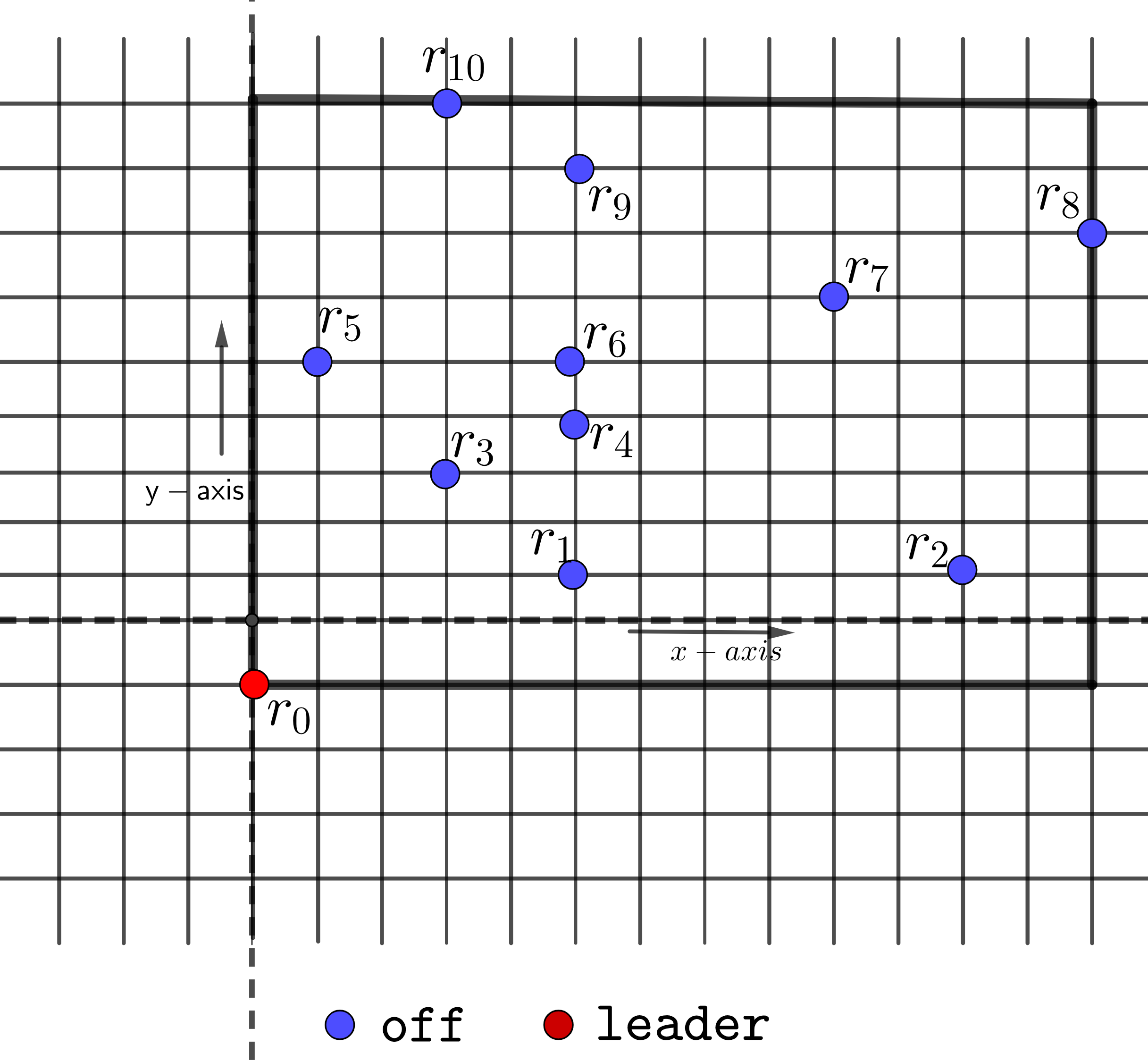}
     \caption{A leader configuration where robot $r_0$ is the robot with light \texttt{leader} at $(0,-1)$ in the agreed coordinate system. }\label{Fig:leaderConfig}
   \end{minipage}
   \hfill
   \begin {minipage}{0.45\textwidth}
    \includegraphics[width=1\linewidth]{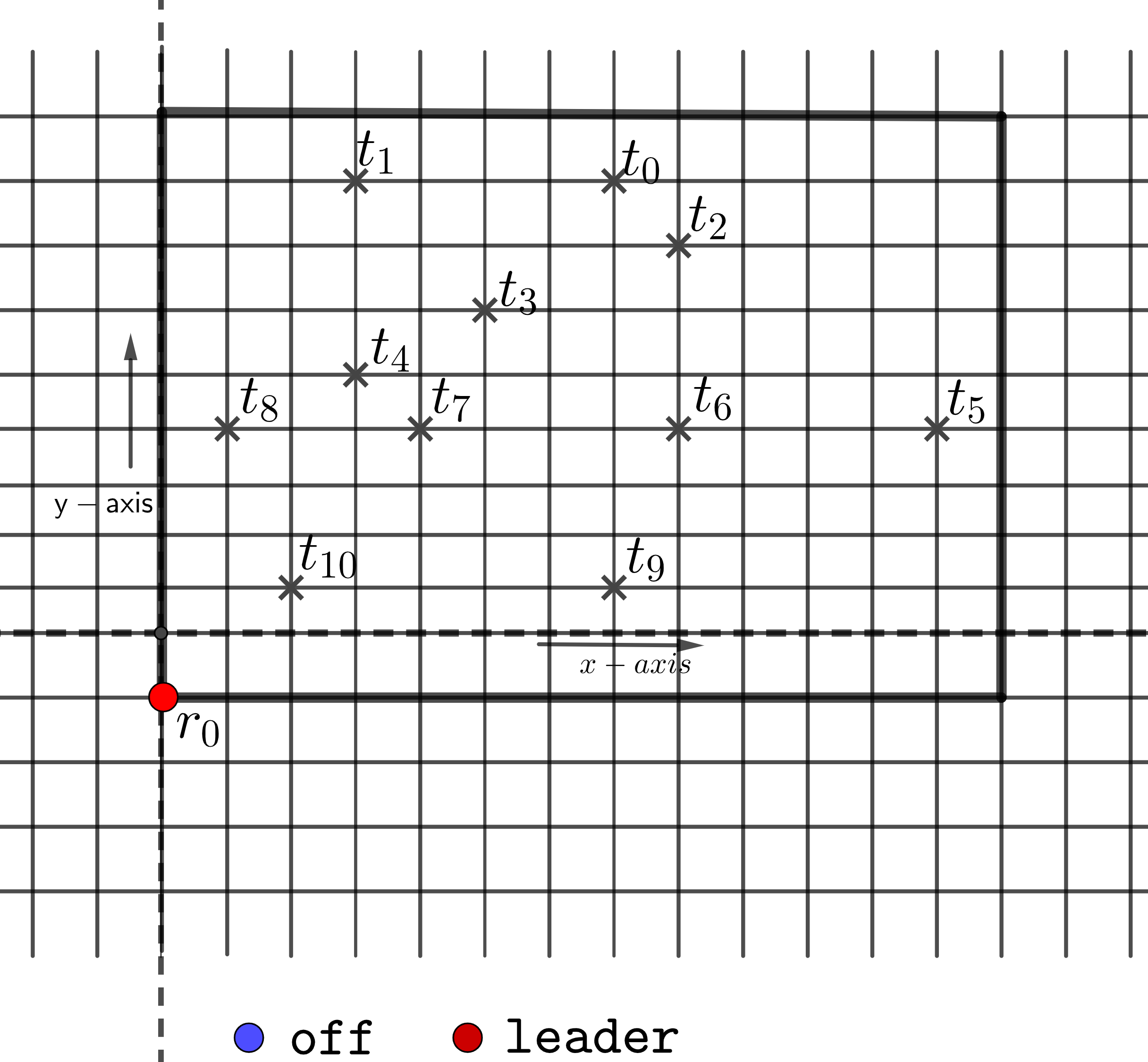}
     \caption{The pattern embedded in the coordinate system.}\label{Fig:TargetEmbedding}
   \end{minipage}
\end{figure}

\begin{figure}[!htb]\centering
   \begin{minipage}{0.45\textwidth}
     \includegraphics[width=1\linewidth]{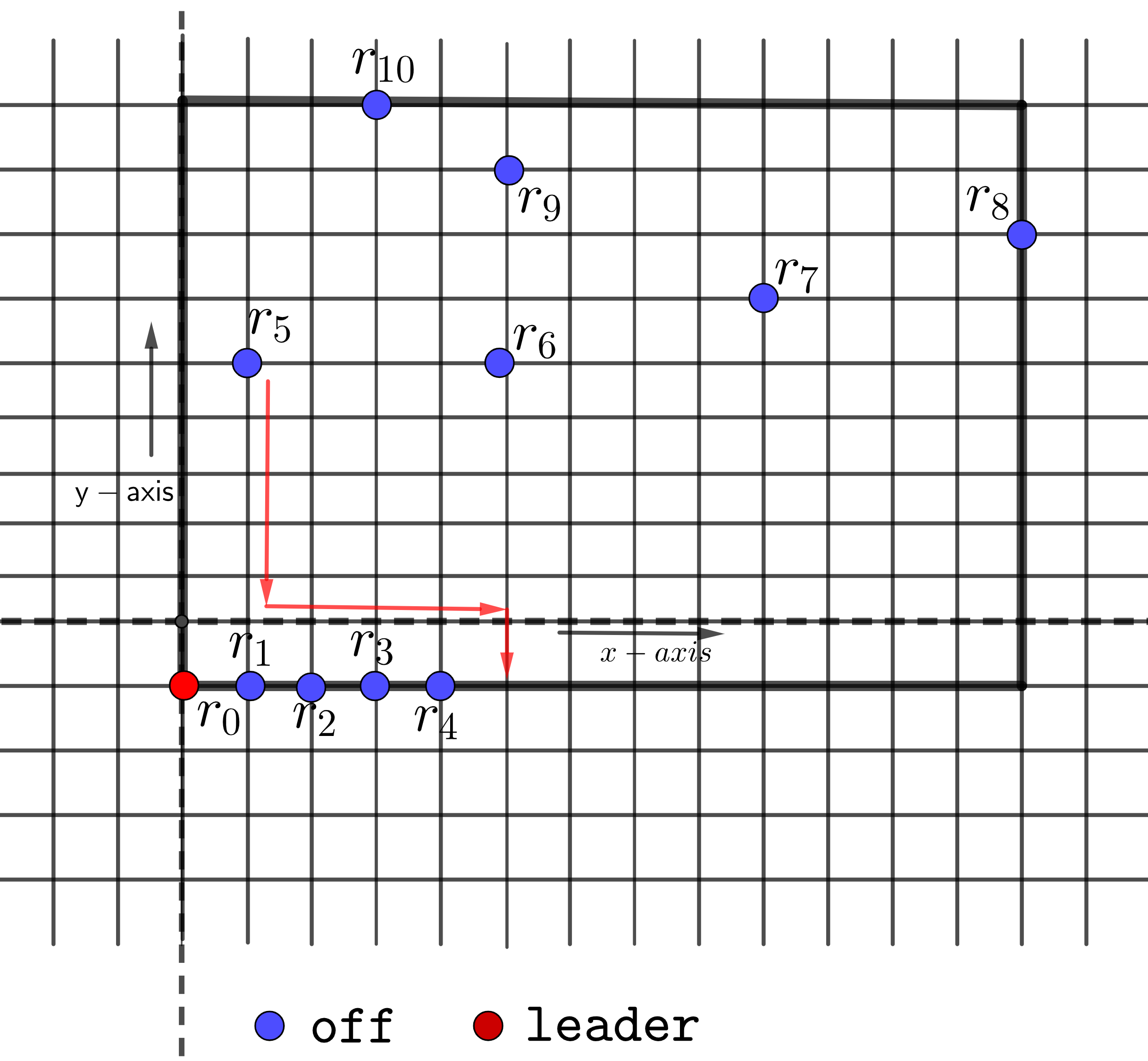}
     \caption{Movement of $r_5$ to $\mathcal{L}_H(r_0)$.}\label{Fig:LineFormation}
   \end{minipage}
   \hfill
   \begin {minipage}{0.45\textwidth}
    \includegraphics[width=1\linewidth]{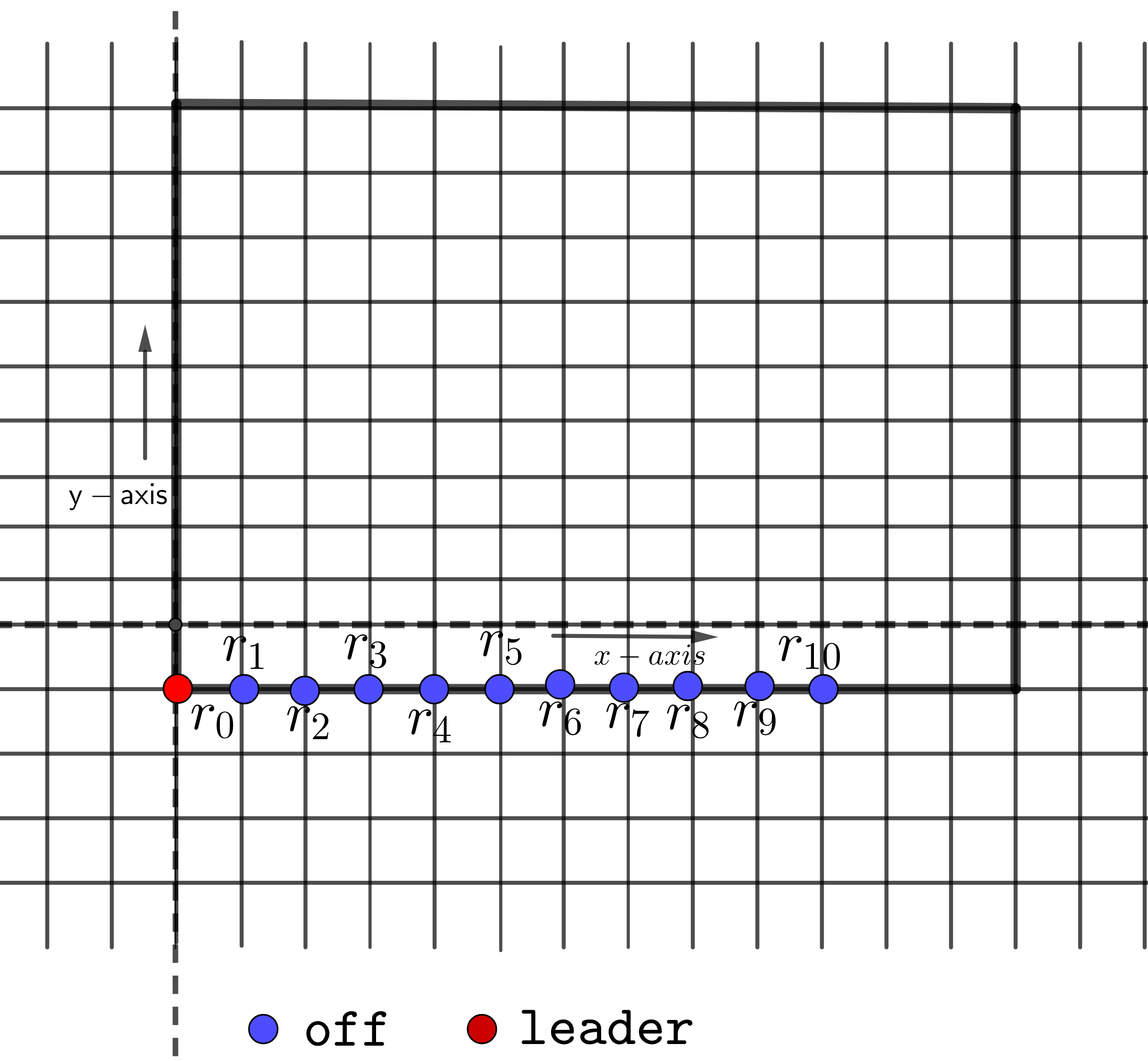}
     \caption{All robots forms an compact line on $\mathcal{L}_H(r_0)$.}\label{Fig:LineFormed}
   \end{minipage}
\end{figure}

\begin{figure}[!htb]\centering
   \begin{minipage}{0.45\textwidth}
     \includegraphics[width=1\linewidth]{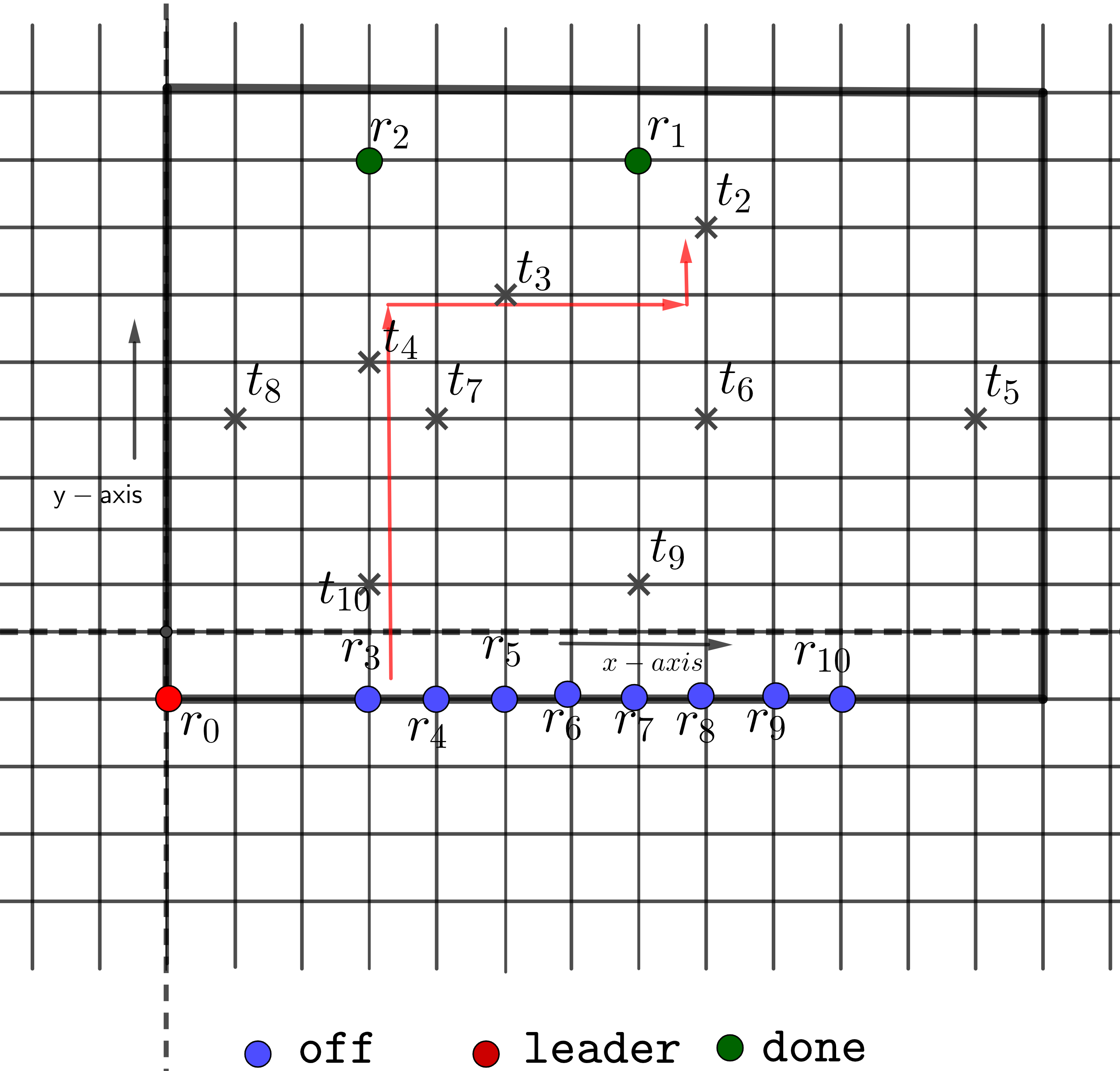}
     \caption{Movement of $r_3$ to $t_2$.}\label{Fig:TargetFormation}
   \end{minipage}
   \hfill
   \begin {minipage}{0.45\textwidth}
    \includegraphics[width=1\linewidth]{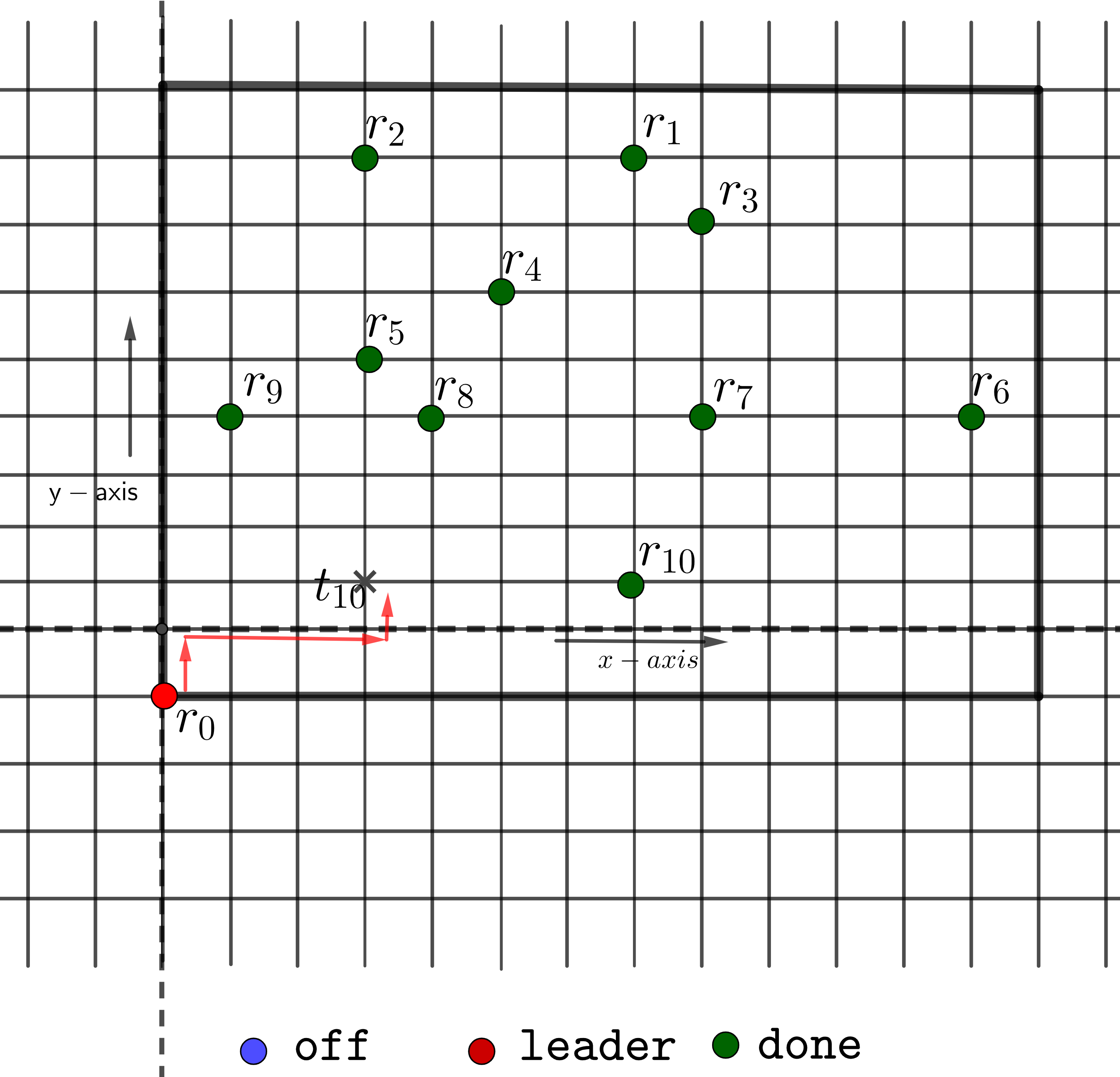}
     \caption{Movement of $r_0$ to $t_{10}$.}\label{Fig:Leader2Target}
   \end{minipage}
\end{figure}

\begin{figure}[ht]
     \includegraphics[width=0.45\linewidth]{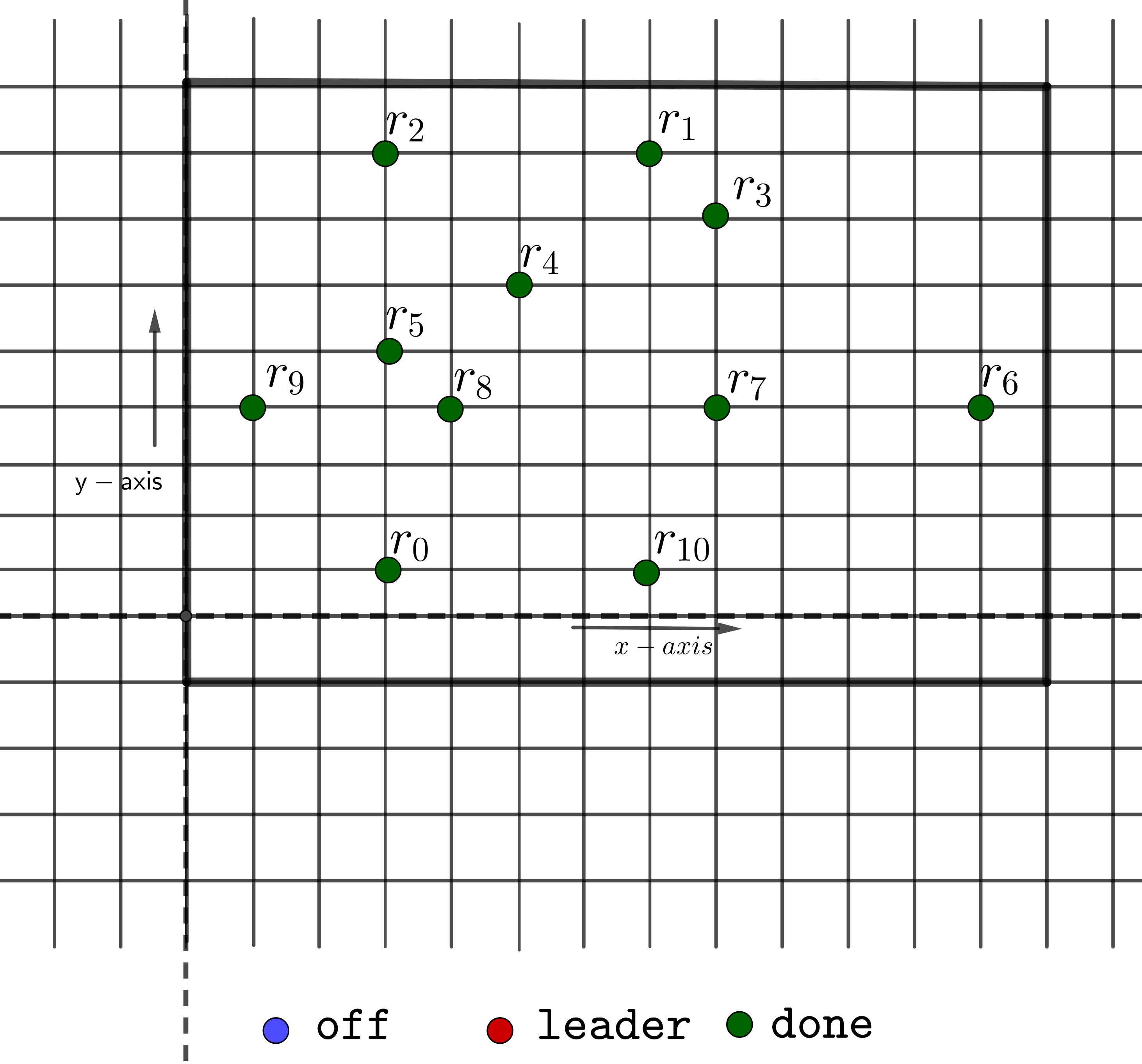}
     \caption{Target formation Achieved.}\label{Fig:Achieved}
\end{figure}

From this information about the number of robots on the line except $r_0$, it can calculate the target position it need to reach and follow the procedure \textsc{TargetMove()} to reach that position. Note that the robot changes its light to \texttt{done} only after reaching its designated target position and in the mean time, no other robot moves from the line while they see robots with light \texttt{off} above them. This technique, of moving the robots sequentially and the ordering of the target co-ordinates, avoid collision in our algorithm. Thus all robots except $r_0$ and $r_{n-1}$ reach their designated target co-ordinates. Now while $r_{n-1}$ moves above from $(n-1,-1)$ to $(n-1,0)$, it sees there is no other robot on the line $\mathcal{L}_H(r_0)$ except $r_0$ and moves to $t_{n-2}$. And finally the robot with light \texttt{leader} moves to the only remaining vacant target position $t_{n-1}$.

If a robot (say, $r$) executes the procedure \textsc{TargetMove($j$)}, that means that the target position of $r$ is $t_j$ with co-ordinate $(t_j(x), t_j(y))$. Observe that, $r$ executes this procedure when it is on $L_{H1}$. Now during this procedure, if $r$ is not on $L_{t_j-1}$, then it moves vertically upwards until it reaches $L_{t_j-1}$. Now, when $r$ is at $L_{t_j-1}$, it checks if the co-ordinate of its current position is $(t_j(x),t_j(y)-1)$. If not, then it moves horizontally to reach the point with the co-ordinate $(t_j(x), t_j(y)-1)$. After that it moves vertically upwards to the point with co-ordinate $(t_j(x), t_j(y))$, which is the target position of $r$ (Figure \ref{Fig:TargetFormation}, \ref{Fig:Leader2Target}, \ref{Fig:Achieved}). Note that during the movement of $r$, there will be no collision as the closed half delimited by $L_{t_j-1}$ and $\mathcal{L}_H(r)$ does not contain any other robot.

Hence from the above discussions, we can conclude the following theorem.

\begin{theorem}
If $\mathbb{C}(T_1)$ is a leader configuration, then $\exists ~T_2 > T_1$ such that $\mathbb{C}(T_2)$ is the target configuration.
\end{theorem}

\section{Conclusion}
Arbitrary pattern formation ($\mathcal{APF})$ has been a very active topic in the field of swarm robotics. It has been thoroughly researched in many different settings. For example, it has been studied when the robots are on a plane or on an infinite grid. Considering obstructed visibility model for robots on a plane, it has been shown that for certain initial configurations,  $\mathcal{APF}$ is solvable with opaque robots having one axis agreement and 6 lights under asynchronous scheduler (\cite{BoseKAS21}). In \cite{abs-1910-02706}, $\mathcal{APF}$ has been solved even with opaque fat robots with light on plane. Comparing to how thoroughly $\mathcal{APF}$ has been studied with robots on plane with obstructed visibility model, it remains quite far behind when it comes to robots on infinite grid with obstructed visibility model. This paper is a stepping stone towards the goal of removing this gap of research regarding $\mathcal{APF}$ between opaque robots on the plane and robots on infinite grid.

In this paper, we have provided a deterministic algorithm for $\mathcal{APF}$ for all solvable initial configurations with one axis agreement and 8 lights under the asynchronous scheduler. For the immediate course of future research, one can think of solving this problem with less numbers of lights. Another interesting way of  extending this problem would be to allow multiplicities in the pattern. 

\paragraph{Acknowledgements.} The second author is supported by UGC, Government of India and the third author is supported by West Bengal State government Fellowship Scheme.

\bibliographystyle{cs-agh}
\bibliography{APF_grid}
\newpage

\end{document}